\PassOptionsToPackage{table}{xcolor}  

\documentclass[runningheads]{llncs}
\usepackage[T1]{fontenc}
\usepackage{amsmath, amssymb,amsfonts} %
\usepackage{graphicx}
\usepackage{textcomp}
\usepackage{xcolor}
\usepackage{algpseudocodex}
\usepackage[ruled, linesnumbered]{algorithm2e}
\usepackage{multirow}
\usepackage{float}
\usepackage{makecell}
\usepackage{hyperref}
\usepackage{cleveref}
\usepackage{tikz}
\usepackage{commath}
\usepackage{siunitx}
\usetikzlibrary{fit,calc,quantikz2}

\usepackage{orcidlink} 

\usepackage{wrapfig}

\usepackage{color}

\usepackage{multirow}
\usepackage{bm}
\newcommand{\tabincell}[2]{\begin{tabular}{@{}#1@{}}#2\end{tabular}}
\usepackage{comment}

\let\oldnl\nl
\newcommand{\nonl}{\renewcommand{\nl}{\let\nl\oldnl}} 

\usepackage{amssymb}

\usepackage[caption=false]{subfig}

\usepackage{enumitem}

\usepackage{booktabs}
\usepackage{pgfplots}
\pgfplotsset{compat=1.18,compat/show suggested version=false}

\usepackage{listings}
\newcommand{\rt}[1]{\textcolor{red}{#1}}

\usepackage[flushleft]{threeparttable}

\renewcommand{\ket}[1]{ \left| #1 \right \rangle}
\renewcommand{\bra}[1]{ \left\langle #1 \right|}

\newcommand{\psucc}{P_{\mathrm{succ}}}
\newcommand{\perr}{P_{\mathrm{err}}}
\newcommand{\pinc}{P_{\mathrm{inc}}}

\newcommand{\supmed}{^{\textnormal{MED}}}
\newcommand{\supmedplus}{^{{\textnormal{MED}}^+}}
\newcommand{\supuqsd}{^{\textnormal{UQSD}}} 

\newcommand{\lSDP}{\lambda}
\newcommand{\leval}{\lambda_{\mathrm{eval}}}
\newcommand{\clb}{c_{\mathrm{lb}}} 

\DeclareMathOperator{\Tr}{Tr} 

\Crefname{equation}{Eq.}{Eqs.}

\crefname{appendix}{appendix}{appendices}
\Crefname{appendix}{Appendix}{Appendices}

\begin{document}

\title{
Error-Tolerant Quantum State Discrimination: Optimization and Quantum Circuit Synthesis
}

\author{
Chien-Kai Ma\inst{2,4,5}\orcidlink{0009-0006-3388-1467}
\and 
Bo-Hung Chen\inst{1,2,4,5}\orcidlink{0000-0003-3709-0225}
\and
Tian-Fu Chen\inst{3,4,5}\orcidlink{0009-0003-5947-6206}
\and
Dah-Wei Chiou\inst{1,2,4,5}\orcidlink{0000-0001-5049-0333}
\and
Jie-Hong~R.~Jiang\inst{1,2,4,5}\orcidlink{0000-0002-2279-4732}
}

\institute{
Department of Electrical Engineering, National Taiwan University, Taiwan 
\and
Graduate Institute of Electronics Engineering, National Taiwan University, Taiwan
\and
Graduate School of Advanced Technology, National Taiwan University, Taiwan 
\and
Center for Quantum Science and Engineering, National Taiwan University, Taiwan 
\and
Physics Division, National Center for Theoretical Sciences, Taiwan
\email{r11943106@ntu.edu.tw, kenny81778189@gmail.com, d11k42001@ntu.edu.tw, dwchiou@gmail.com, jhjiang@ntu.edu.tw}
}

\titlerunning{Error-tolerant quantum state discrimination}
\authorrunning{C.-K. Ma, B.-H. Chen, T.-F. Chen, D.-W. Chiou, J.-H. R. Jiang }

\maketitle


\begin{abstract}
We develop \emph{error-tolerant} quantum state discrimination\linebreak (QSD) strategies that maintain reliable performance under moderate noise. Two complementary approaches are proposed: \emph{CrossQSD}, which generalizes unambiguous discrimination with tunable confidence bounds to balance accuracy and efficiency, and \emph{FitQSD}, which optimizes the measurement outcome distribution to approximate that of the ideal noiseless case. Furthermore, we provide a unified \emph{hybrid-objective QSD} framework that continuously interpolates between minimum-error discrimination (MED) and FitQSD, allowing flexible trade-offs among competing objectives. The associated optimization problems are formulated as convex programs and efficiently solved via disciplined convex programming or, in many cases, semidefinite programming. Additionally, a circuit synthesis framework based on a modified Naimark dilation and isometry synthesis enables hardware-efficient implementations with substantially reduced qubit and gate resources. An open-source toolkit automates the full optimization and synthesis workflow, providing a practical route to QSD on current quantum devices.
\end{abstract}


\section{Introduction}
\emph{Quantum state discrimination} (QSD) \cite{QSD1,QSD2,QSD4} is a fundamental problem in quantum information science that concerns how to identify, as accurately and efficiently as possible, the state of a quantum system known to be prepared in one of several predefined states. This task underpins many areas of quantum technology, including quantum cryptography, quantum communication, and quantum computing \cite{bb84,RN492}. Unlike classical states, quantum states can be nonorthogonal and therefore cannot be perfectly distinguished by any physical measurement. This intrinsic indistinguishability is not a technological limitation but a fundamental feature of quantum mechanics. It plays a pivotal role in the security of quantum key distribution protocols such as BB84 \cite{bb84} and SARG04 \cite{RN492}, where information is encoded in nonorthogonal states so that any eavesdropping attempt inevitably disturbs the system and can thus be detected.

Over the years, several strategies have been developed to address QSD, each offering a distinct trade-off between accuracy, confidence, and efficiency. The most widely studied approaches are \emph{minimum-error discrimination} (MED) \cite{QSD1,MED1}, which minimizes the average probability of incorrectly identifying a state, and \emph{unambiguous quantum state discrimination} (UQSD) \cite{QSD4,QSD3}, which permits inconclusive outcomes to ensure that every conclusive result is error-free. UQSD, for instance, plays a key role in the B92 protocol \cite{RN485}, where each received state must be either reliably identified or discarded.
Beyond these, other strategies pursue different optimization objectives. \emph{Maximum-confidence discrimination} \cite{RN519} aims to maximize the confidence that a conclusive outcome correctly identifies the state, even at the cost of allowing occasional errors or inconclusive results, while \emph{worst-case a posteriori probability discrimination} \cite{kosut2004quantumstatedetectordesign} seeks to maximize the worst-case posterior probability of being correct, ensuring robust performance in the least favorable cases. Together, these approaches illustrate the diverse operational perspectives and optimization criteria that underlie the general problem of quantum state discrimination.

In realistic settings, however, unavoidable noise introduces additional errors and can significantly undermine the performance of conventional strategies, which are typically formulated under ideal, noiseless assumptions. For example, when depolarizing noise affects all predefined states, no matter how small the noise is, UQSD will always yield an inconclusive outcome, rendering the method ineffectual. To address such challenges, the \emph{fixed rate of inconclusive outcome} (FRIO) method \cite{FRIO1,FRIO2} has been proposed as a compromise between MED and UQSD: it relaxes the guarantee of conclusive correctness while enforcing an upper bound on the rate of inconclusive outcomes.

In this work, we further tackle the issue of noise from multiple perspectives by developing \emph{noise-aware} approaches that maintain reliable performance under realistic noisy conditions. Specifically, we propose two novel error-tolerant QSD strategies, along with some auxiliary methods that account for different aspects of error tolerance. The first strategy, \emph{CrossQSD}, generalizes UQSD by incorporating tunable confidence bounds associated with false positive and false negative errors,\footnote{A false negative with respect to a specific predefined state $\rho_i$ refers to the event that the given state is $\rho_i$ but it is not identified as $\rho_i$, i.e., the outcome corresponds to $\rho_j$ with $j \neq i$. Conversely, a false positive for $\rho_i$ refers to the event that the outcome identifies the state as $\rho_i$ while the given state is actually $\rho_j$ with $j \neq i$.}
thereby enabling a more controlled balance between accuracy, confidence, and efficiency. The second strategy, \emph{FitQSD}, optimizes the probability distribution of measurement outcomes to reproduce, as closely as possible, that of the ideal noiseless UQSD case, thus preserving discrimination performance even in the presence of noise. Furthermore, MED and FitQSD can be integrated into a unified \emph{hybrid-objective QSD} framework, offering a continuous interpolation between MED and UQSD and enabling a flexible trade-off among multiple optimization objectives under noisy conditions.\footnote{Error-tolerant QSD is a relatively new and underexplored topic. To the best of our knowledge, the only prior approach that directly addresses noise effects is the FRIO framework \cite{FRIO1,FRIO2}. As practical methods, our approaches should therefore be viewed as complementary to FRIO rather than directly comparable. A fair and meaningful comparison would require a broader context---e.g., in quantum communication or cryptography---which lies beyond the scope of this work and warrants further investigation. Our methods, together with FRIO, can serve as tools for such future studies.}

The optimization problems underlying both CrossQSD and FitQSD, as well as all auxiliary and hybrid methods, can be elegantly expressed within the general framework of convex optimization \cite{boyd2004convex,MED3} and efficiently solved using disciplined convex programming (DCP) \cite{grant2006disciplined} or its parametric extension, disciplined parametrized programming (DPP) \cite{agrawal2019differentiable}. In many instances, these formulations further reduce to semidefinite programming (SDP) \cite{vandenberghe1996semidefinite,RN23}, which allows even more efficient numerical solutions.
This formulation naturally yields the optimal positive operator-valued measure (POVM) and readily extends to cases involving both mixed and pure predefined states, thereby enabling a systematic investigation of noise effects on discrimination performance and the comparative advantages of different approaches under various conditions.

Even when the optimal POVM for a given QSD strategy is known, its physical realization on quantum hardware is often nontrivial. Although quantum circuit synthesis for POVMs has been studied, a general synthesis flow tailored for QSD problems has been lacking. To bridge this gap, as our second main contribution, we develop a circuit synthesis framework based on a modified version of Naimark's dilation theorem \cite{naimark1943representations} and the state-of-the-art isometry circuit  synthesis \cite{RN12}, enabling efficient implementation of QSD strategies on quantum circuits. For specific cases, such as coherent states relevant to optical quantum systems \cite{gazeau2009coherent,arecchi1972atomic}, we demonstrate that our approach needs fewer qubits and can substantially reduce gate counts and circuit depth, with only a modest trade-off in discrimination success probability.

Although our primary focus is on the discrimination of multi-qubit states within quantum-circuit architectures, the constructed circuits faithfully capture the essential information flow of the measurement process and can be adapted to other physical platforms where basic quantum gates are available.
To support practical implementations, we also provide an open-source toolkit that automates the optimization and synthesis of the proposed QSD strategies. With these tools and methodologies, our work establishes a systematic foundation for performing practical QSD on current and emerging quantum hardware.

\section{Error-Tolerant Quantum State Discrimination}
\label{sec:error_uqsd}


The problem of QSD can be formulated in its most general form as follows.  
Given a predefined set of $k$ quantum states $\{\rho_{1}, \rho_{2}, \dots, \rho_{k}\}$ (each being either a pure or mixed state), Alice prepares one of them with prior probabilities $\{p_1, p_2, \dots, p_k\}$ satisfying $\sum_{i=1}^{k} p_i = 1$, and sends it to Bob through a quantum channel $\mathcal{E}: \rho_i \mapsto \mathcal{E}(\rho_i) = \rho'_i$, which may or may not introduce noise. Upon receiving $\rho'_i$, Bob performs a POVM measurement characterized by a POVM $\{\Pi_{1}, \Pi_{2}, \dots, \Pi_{k}\}$, where each element $\Pi_i$ is a positive-definite operator corresponding to the outcome by which Bob guesses the prepared state to be $\rho_i$. Alternatively, if an inconclusive outcome is allowed, the measurement is extended to $\{\Pi_{1}, \Pi_{2}, \dots, \Pi_{k}, \Pi_{k+1} \equiv \Pi_{?}\}$, where $\Pi_{?}$ represents the operator associated with the inconclusive result.  
The problem of QSD is to solve the optimal POVM, by which Bob can identify the prepared state as accurately and efficiently as possible under various intended conditions or considerations.

Our toolkit supports multiple types of optimization strategies and is readily extensible to any strategy that can be formulated as a convex optimization problem with constraints expressed in terms of POVM operators. In this section, we formulate the CrossQSD and FitQSD methods, together with other modified QSD schemes, each targeting different optimization objectives. We demonstrate the efficacy and robustness of these schemes on predefined sets of nearly orthogonal states and truncated coherent states. For simplicity, noise in our experiments is modeled by a depolarizing channel
\begin{equation}
\mathcal{E}_\lambda(\rho): \rho \mapsto (1-\lambda)\rho + \lambda \frac{I}{d},
\end{equation}
where $\lambda$ denotes the noise level and $d$ is the Hilbert-space dimension.

\subsection{Modified Minimum Error Discrimination (MED${}^+$)}

We first introduce MED$^+$, a straightforward extension of the conventional MED scheme that allows an inconclusive outcome, similar to UQSD. The optimization problem can be cast as a semidefinite program:
\begin{subequations}
\begin{align}
      \underset{\{\Pi_i\}}{\text{maximize}}\qquad   & \psucc = \sum_{i=1}^{k} p_{i} \operatorname{Tr}(\rho_{i} \Pi_{i}), \\
      \text{subject to} \qquad                    & \Pi_{i} \succeq 0, \quad i = 1, \ldots, k+1, \\
                                            & \sum_{i=1}^{k+1} \Pi_i = I,
\end{align}
\end{subequations}
where the optimization objective is the ``success probability'' $\psucc$---the overall probability of correctly identifying the prepared state---and $A \succeq B$ indicates that $A-B$ is positive semidefinite.

Although extending MED to MED${}^+$ intuitively enlarges the feasible space and appears to enhance optimality, it in fact provides no additional benefit, as established in the following theorem, whose proof is provided in \Cref{sec:proof-1}.

\begin{theorem} \label{thm:med_medplus}
Let both MED and MED${}^+$ be applied under the same predefined states $\{\rho_{1}, \rho_{2}, \dots, \rho_{k}\}$ and prior probabilities $\{p_{1}, p_{2}, \dots, p_{k}\}$. Then, in the optimal MED${}^+$ solution, the inconclusive operator necessarily vanishes, i.e., $\Pi_? = 0$, rendering the optimal MED${}^+$ identical to the optimal MED. In particular, the optimal success probabilities of the two schemes are equal, i.e.,
\begin{equation}
(\psucc)\supmed = (\psucc)\supmedplus.
\end{equation}
\end{theorem}

Despite yielding no improvement in optimality, the distinction between MED and MED${}^+$ offers valuable conceptual insight. Specifically, \Cref{thm:med_medplus} immediately implies the following corollary, whose proof is provided in \Cref{sec:proof-2}.

\begin{corollary}
\label{corollary}
For the same predefined states and prior probabilities, the optimal success probability of MED is always greater than or equal to that of UQSD; that is,
\begin{equation}
(\psucc)\supmed \geq (\psucc)\supuqsd.
\end{equation}
\end{corollary}

This result highlights a fundamental distinction between the two strategies: MED allows a nonzero probability of error to maximize the overall success rate, whereas UQSD strictly forbids any incorrect identification and therefore achieves a smaller (or equal) success probability. In other words, UQSD prioritizes certainty over completeness, while MED optimizes performance under uncertainty.

\subsection{CrossQSD}
\label{sec:crossqsd}
To enable controlled tolerance of both false positive and false negative  errors within user-specified limits, we propose \emph{CrossQSD}, a relaxation of UQSD formulated as a semidefinite program:
\begin{subequations}
\begin{align}
    \underset{\{\Pi_i\}}{\text{maximize}} \qquad & 
        \psucc = \sum_{i=1}^{k} p_i \operatorname{Tr}\!\left[\mathcal{E}_{\lambda}(\rho_i) \Pi_i\right], \label{eq:crossqsd_obj} \\[3pt]
    \text{subject to} \qquad
        & \Pi_i \succeq 0, \quad i = 0, \ldots, k+1, \label{eq:crossqsd_psd} \\
        & \sum_{i=0}^{k+1} \Pi_i = I, \label{eq:crossqsd_sum} \\
        & p(\Pi_i \mid \rho_i) = \frac{p_i\, \operatorname{Tr}\!\left[\mathcal{E}_{\lambda}(\rho_i) \Pi_i\right]}
        {\sum_{j=1}^k p_i\, \operatorname{Tr}\!\left[\mathcal{E}_{\lambda}(\rho_i) \Pi_j\right]}
        \geq 1-\alpha_i, \quad i = 1, \ldots, k, \label{eq:crossqsd_type1} \\
        & p(\rho_i \mid \Pi_i) = \frac{p_i\, \operatorname{Tr}\!\left[\mathcal{E}_{\lambda}(\rho_i) \Pi_i\right]}
        {\sum_{j=1}^k p_j\, \operatorname{Tr}\!\left[\mathcal{E}_{\lambda}(\rho_j) \Pi_i\right]}
        \geq 1-\beta_i, \quad i = 1, \ldots, k, \label{eq:crossqsd_type2}
\end{align}
\end{subequations}
where $\alpha_i$ and $\beta_i$ are user-defined upper bounds on the false positive and false negative error tolerances for identifying $\rho_i$, and $\lambda$ denotes the noise level.  

CrossQSD provides a tunable compromise between accuracy, confidence, and efficiency. It is conceptually related to maximum-confidence discrimination~\cite{RN519} and worst-case a posteriori probability discrimination~\cite{kosut2004quantumstatedetectordesign}, but extends these schemes in two crucial ways: first, by explicitly taking noise into account via $\mathcal{E}_\lambda$; and second, by allowing simultaneous, independent control over both false positive and false negative error tolerances through the adjustable parameters $\alpha_i$ and $\beta_i$. This flexibility enables a continuous interpolation between the regimes of strict unambiguity and minimal error, offering a unified framework for QSD in the presence of noise.

To evaluate the performance of CrossQSD, we consider three coherent states truncated to an eight-dimensional Hilbert space (corresponding to three qubits) as the predefined noise-free states $\rho_i$, each assigned an equal prior probability $p_i = 1/3$. A coherent state $\ket{\alpha}$ is defined as the eigenstate of the annihilation operator $a$, satisfying $a\ket{\alpha} = \alpha\ket{\alpha}$. Specifically, we choose $\alpha = 1, e^{-i\pi/3}$, and $e^{-i2\pi/3}$. The tolerance bounds are set to $\alpha_i = \beta_i = 0.01$.

The noise level $\lambda$ is varied from $10^{-6}$ to $10^0$ to simulate different degrees of noise contamination. However, when solving the SDP optimization, we assume a \emph{fixed} evaluation noise level of $\leval = 0.01$ for $\lambda$, in order to deliberately examine the effect of a mismatch between the anticipated and actual noise.

The results, shown in \Cref{fig:crossqsd_results}, are presented in terms of the error-to-success ratio $\perr / \psucc \equiv (1 - \psucc)/\psucc$. The error-to-success ratio remains small as long as $\lambda$ is below $\leval$, whereas it increases significantly when $\lambda$ becomes much larger than $\leval$.
These findings demonstrate that, as long as the noise level is not underestimated and the error tolerances $\alpha_i$ and $\beta_i$ are chosen to be reasonably comparable to $\leval$, CrossQSD delivers consistently robust performance, even when the noise level is significantly overestimated due to incomplete knowledge of its characteristics.

\begin{figure}[ht]
    \centering
    \resizebox{0.93\textwidth}{!}{
        \begin{tikzpicture}
            \begin{axis}[
                    xlabel={Noise Level \(\lambda\)},
                    ylabel={Error-to-Success Ratio \(\perr / \psucc\)},
                    xmode=log,
                    ymode=log,
                    xmin=1e-6, xmax=1,
                    ymin=1e-3, ymax=10,
                    grid=major,
                    width=0.8\textwidth,
                    height=5.5cm,
                    legend pos=north west
                ]
                \addplot[mark=*, blue] coordinates {
                    (1.0000e-06, 0.00511489) (1.77828e-06, 0.00511527) (3.16228e-06, 0.00511596)
                    (5.62341e-06, 0.00511718) (1.0000e-05, 0.00511934) (1.77828e-05, 0.00512320)
                    (3.16228e-05, 0.00513005) (5.62341e-05, 0.00514223) (1.0000e-04, 0.00516389)
                    (1.77828e-04, 0.00520242) (3.16228e-04, 0.00527094) (5.62341e-04, 0.00539282)
                    (1.0000e-03, 0.00560968) (1.77828e-03, 0.00599567) (3.16228e-03, 0.00668318)
                    (5.62341e-03, 0.00790931) (1.77828e-02, 0.0140346)
                    (3.16228e-02, 0.0211461) (5.62341e-02, 0.0341741) (1.0000e-01, 0.0586295)
                    (1.77828e-01, 0.106708) (5.62341e-01, 0.487311) (1.0000e+00, 2.00000)
                };
                \addlegendentry{CrossQSD}
                \addplot[mark=square*, red, mark size=3pt] coordinates {
                    (1.0000e-02, 0.0101010)
                };
                \addlegendentry{\(\leval = 10^{-2}\)}
            \end{axis}
        \end{tikzpicture}
    } 
    \caption{CrossQSD: Error-to-success ratio \(\perr / \psucc\) vs.\ noise level $\lambda$. The predefined states are three coherent states truncated to three qubits; \(\alpha_i = \beta_i = 0.01\); $\leval=0.01$.}
    \label{fig:crossqsd_results}
\end{figure}
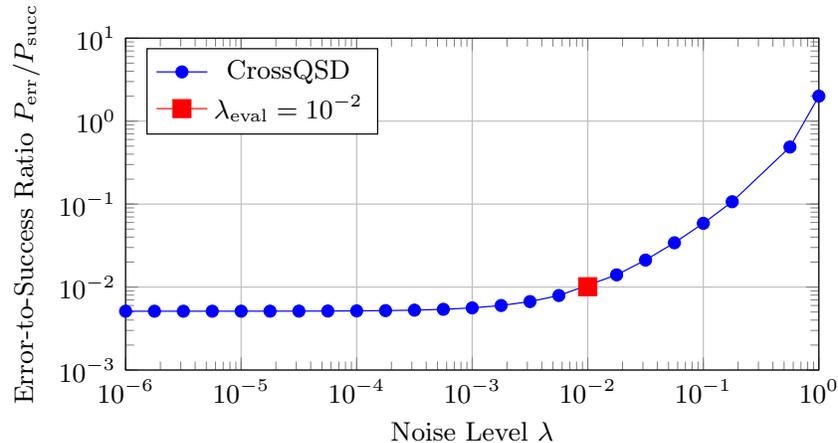

\subsection{FitQSD}
The other direction for extending UQSD to handle noise is the \emph{FitQSD} approach, which seeks to make the probability distribution of measurement outcomes under noise as close as possible to that of the corresponding noiseless UQSD, according to a chosen notion of ``closeness.''

Specifically, consider the joint probability that Alice sends $\rho_i$ while Bob identifies it as $\rho_j$, given by
\begin{equation}
p(\rho_i, \Pi_j) = p_i \Tr(\rho_i \Pi_j).
\end{equation}
Let $\{\Pi_i^0\}$ denote the optimal POVM for the noiseless UQSD, yielding the reference probability distribution $\{p(\rho_i, \Pi^0_j)\}$ for all $i,j = 1, \dots, k+1$. In the presence of noise, described by the quantum channel $\mathcal{E}_\lambda$, a QSD strategy with POVM $\{\Pi_i\}$ instead produces $\{p(\mathcal{E}\lambda(\rho_i), \Pi_j)\}$. The goal of FitQSD is to make these two probability distributions as close as possible, given that the reference $\{p(\rho_i,\Pi_j^0)\}$ is known or solved in advance.

The closeness can be quantified by an $L_p$-norm, leading to the following convex optimization formulation: \allowdisplaybreaks
\begin{subequations}
\begin{align}
 \underset{\{\Pi_i\}}{\text{minimize}}\qquad   & \sum_{i,j=1}^{k+1} \left| p(\rho_i,\Pi^0_j) - p(\mathcal{E}_\lambda(\rho_i),\Pi_j) \right|^\ell,                            \\ 
 \text{subject to}\qquad                     & \Pi_{i}\succeq 0, \quad i = 0, \ldots, k+1, \\
                                       & \sum_{i=0}^{k+1} \Pi_i = I,
\end{align}
\end{subequations} \allowdisplaybreaks
where $\ell>0$ corresponds to the index $p$ of the $L_p$-norm.
Particularly, the cases with $\ell = 1$ and $\ell = 2$ are referred to as \emph{FitQSD-MinL1} (minimizing the $L_1$ distance) and \emph{FitQSD-MinSS} (minimizing the sum of squares), respectively.

Alternatively, rather than minimizing an $L_p$-norm directly, one may instead maximize the success probability---similar to CrossQSD---while still maintaining a close match between the noisy and noiseless cases. This leads to the \emph{FitQSD-MECO} (minimizing error with constrained overlap) formulation, expressed as a semidefinite program:
\begin{subequations}
\begin{align}
\label{eq:maxpsucc_mindiff_obj}
& \underset{\{\Pi_i\}}{\text{maximize}} 
&  & \psucc = \sum_{i=1}^{k} p_{i} \Tr(\mathcal{E}_{\lambda}(\rho_{i}) \Pi_{i}),\\
& \text{subject to}                   
&  & \Pi_{i} \succeq 0, \quad i = 1, \ldots, k+1,\\
&                                     
&  & \sum_{i=1}^{k+1} \Pi_i = I,\\
&                                     
&  & p_{i} \operatorname{Tr}(\mathcal{E}_{\lambda}(\rho_{i}), \Pi_{i}) \leq p(\rho_i, \Pi^0_i), \quad \text{for } i=1,\ldots,k,\\
&                                     
&  & p_{i} \operatorname{Tr}(\mathcal{E}_{\lambda}(\rho_{i}), \Pi_{j}) \geq p(\rho_i, \Pi^0_j), \quad \text{for } i \neq j, \; i, j =1,\ldots,k.
\end{align}
\end{subequations}
Here, the reference probabilities $p(\rho_i, \Pi^0_i)$ act as upper bounds that prevent the optimization from overly boosting certain correct-identification rates\linebreak $p_i \Tr(\mathcal{E}{\lambda}(\rho_i) \Pi_i)$ at the expense of others.
Similarly, the reference probabilities $p(\rho_i, \Pi^0_j)$ for $i \neq j$ serve as lower bounds that prevent the optimization from excessively suppressing specific misidentification rates $p_i \Tr(\mathcal{E}{\lambda}(\rho_i) \Pi_j)$ in favor of others.
These comparative bounds constrain the optimized joint probability distribution $\{p(\mathcal{E}_{\lambda}(\rho_i), \Pi_j)\}$ to remain close to its noiseless reference $\{p(\rho_i, \Pi^0_j)\}$, thereby ensuring consistency and robustness.

To evaluate the performance of FitQSD, we consider the following three two-qubit entangled states with equal prior probabilities:
\begin{equation}\label{eq:2-qubits states}
\ket{\psi_i} = \frac{1}{\sqrt{1 + a_i^2}} \left( \ket{b_{i-1}} + a_i \ket{11} \right), \quad \mathrm{for}\quad i=1,2,3,
\end{equation}
where $a_1 = 0.2$, $a_2 = 0.5$, $a_3 = 0.7$, and $b_i$ denotes the binary representation of $i$. The noise level $\lambda$ is varied from $10^{-6}$ to $10^{-1}$ to simulate different noise regimes. We use MOSEK~\cite{mosek} 
to solve the convex optimization.

The results are presented in \Cref{fig:fitqsd_results_subfigs}. To ensure a fair comparison, the closeness between the noisy and noiseless distributions is consistently quantified using the $L_2$ distance across all three methods. FitQSD-MinL1 and FitQSD-MECO yield nearly identical success probabilities and $L_2$ distances, whereas FitQSD-MinSS exhibits noticeable deviations in both aspects. All three methods maintain high success probabilities as long as the noise level $\lambda$ remains relatively low. As $\lambda$ increases, their performance in $\psucc$ gradually deteriorates, with MinSS showing a milder decline. Conversely, the $L_2$ distance increases for all methods as $\lambda$ rises, but MinSS performs noticeably worse when $\lambda$ is small.
Therefore, if both success probability and distributional closeness between the noisy and noiseless measurement outcomes are to be optimized, MinL1 and MECO are preferable in the low-noise regime, whereas MinSS performs better under relatively high noise.

It is intriguing---and warrants further investigation---that MinL1 and MECO, though conceptually quite different, yield remarkably similar results, whereas MinL1 and MinSS, which are conceptually very similar, produce noticeably different outcomes. The close agreement between MinL1 and MECO is particularly valuable, as the latter can serve as a practical substitute for the former: MECO is formulated as a semidefinite program and can therefore be solved much more efficiently while achieving nearly the same performance. Furthermore, the FitQSD framework can be naturally extended to other values of $\ell$ in the $L_p$-norm, which may provide additional advantages under different noise characteristics---a direction that deserves further exploration.




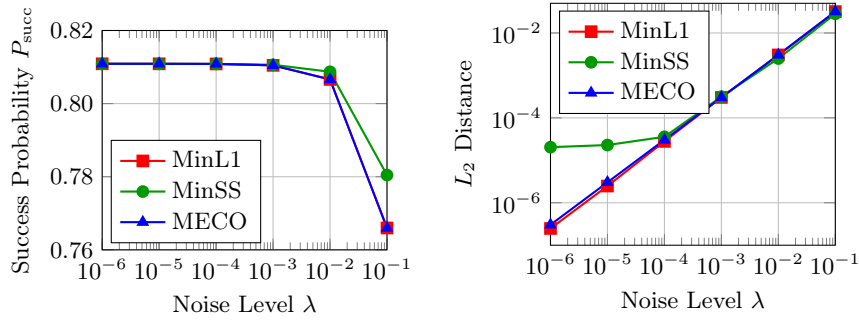
\begin{figure}[t]
  \centering
  \begin{minipage}{0.44\textwidth}
    \centering
    \begin{tikzpicture}
      \begin{axis}[
        xlabel={Noise Level $\lambda$},
        ylabel={Success Probability $\psucc$},
        xmode=log, xmin=1e-6, xmax=1e-1,
        ymin=0.76, ymax=0.82,
        width=\linewidth, height=4.5cm,
        grid=major,                
        legend pos=south west,
        legend cell align=left,
        yticklabel style={
          /pgf/number format/fixed,
          /pgf/number format/precision=2,
          /pgf/number format/fixed zerofill,
        },
      ]
        \node[anchor=south west] at (rel axis cs:0.01,0.98) {\small (a) $\psucc$};

        \addplot[red, thick, mark=square*] coordinates {
          (1.00e-06, 0.8108929) (1.00e-05, 0.8108909) (1.00e-04, 0.8108562)
          (1.00e-03, 0.8104668) (1.00e-02, 0.8066100) (1.00e-01, 0.7659851)
        };
        \addlegendentry{MinL1}

        \addplot[green!60!black, thick, mark=*] coordinates {
          (1.00e-06, 0.8108753) (1.00e-05, 0.8108731) (1.00e-04, 0.8108618)
          (1.00e-03, 0.8106197) (1.00e-02, 0.8087073) (1.00e-01, 0.7804561)
        };
        \addlegendentry{MinSS}

        \addplot[blue, thick, mark=triangle*] coordinates {
          (1.00e-06, 0.8108927) (1.00e-05, 0.8108889) (1.00e-04, 0.8108504)
          (1.00e-03, 0.8104660) (1.00e-02, 0.8066098) (1.00e-01, 0.7659851)
        };
        \addlegendentry{MECO}
      \end{axis}
    \end{tikzpicture}
  \end{minipage}
  \hspace{0.025\textwidth} 
  \begin{minipage}{0.44\textwidth}
    \centering
    \begin{tikzpicture}
      \begin{axis}[
        xlabel={Noise Level $\lambda$},
        ylabel={$L_2$ Distance},
        xmode=log, xmin=1e-6, xmax=1e-1,
        ymode=log, ymin=1e-7, ymax=5e-2,
        width=\linewidth, height=4.8cm,
        grid=major,                
        legend pos=north west,
        legend cell align=left,
        yticklabel style={
          /pgf/number format/sci,
          /pgf/number format/precision=1,
        },
      ]
        \node[anchor=south west] at (rel axis cs:0.01,0.98) {\small (b) $L_2$};

        \addplot[red, thick, mark=square*] coordinates {
          (1.00e-06, 2.48e-07) (1.00e-05, 2.48e-06) (1.00e-04, 2.79e-05)
          (1.00e-03, 3.03e-04) (1.00e-02, 3.04e-03) (1.00e-01, 3.19e-02)
        };
        \addlegendentry{MinL1}

        \addplot[green!60!black, thick, mark=*] coordinates {
          (1.00e-06, 2.04e-05) (1.00e-05, 2.29e-05) (1.00e-04, 3.58e-05)
          (1.00e-03, 3.16e-04) (1.00e-02, 2.50e-03) (1.00e-01, 2.83e-02)
        };
        \addlegendentry{MinSS}

        \addplot[blue, thick, mark=triangle*] coordinates {
          (1.00e-06, 3.03e-07) (1.00e-05, 3.03e-06) (1.00e-04, 3.03e-05)
          (1.00e-03, 3.03e-04) (1.00e-02, 3.04e-03) (1.00e-01, 3.19e-02)
        };
        \addlegendentry{MECO}
      \end{axis}
    \end{tikzpicture}
  \end{minipage}

  \caption{FitQSD: Success probability $\psucc$ (left) and $L_2$ distance (right) vs.\ noise level $\lambda$. The predefined states are given by \Cref{eq:2-qubits states}. The MOSEK solver precision is set to $10^{-9}$.}
  \label{fig:fitqsd_results_subfigs}
\end{figure}

\subsection{Hybrid-objective QSD} 
In addition to the previously discussed variants, we introduce the \emph{hybrid-objective QSD} scheme, which combines multiple optimization objectives to achieve a balanced trade-off among accuracy, performance, and robustness. This hybrid formulation provides a unified and flexible framework that encompasses and extends conventional approaches such as MED, UQSD, and their noise-aware generalizations. By jointly maximizing the success probability and minimizing the difference between the noisy and noiseless measurement outcomes---with a tunable parameter $w$ adjusting the trade-off between these two competing objectives---the hybrid-objective QSD enables a continuous interpolation between MED${}^+$ and FitQSD.

The corresponding semidefinite program is formulated, with $w\geq0$ and $\ell\geq0$, as \allowdisplaybreaks
\begin{subequations}
\begin{align}
\label{eq:roland_obj}
& \underset{\{\Pi_i\}}{\text{maximize}} 
&  & \sum_{i=1}^{k} p_{i} \Tr(\mathcal{E}_{\lambda}(\rho_{i}) \Pi_{i}) 
 - w \sum_{i,j =1}^{k+1} \left| p(\rho_i,\Pi^0_j) - p(\mathcal{E}_{\lambda}(\rho_i),\Pi_j) \right|^\ell,                      
\\
& \text{subject to}                   
&  & \Pi_{i} \succeq 0, \quad i = 1, \ldots, k+1,
\\
&                                     
&  & \sum_{i=1}^{k+1} \Pi_i = I.
\end{align}
\end{subequations}
Equivalently, this continuous interpolation can be viewed as one between MED ($w \to 0$) and UQSD ($w \to \infty$), since MED${}^+$ has been shown to be equivalent to MED, while FitQSD is designed to reproduce the measurement outcomes of the noiseless UQSD as closely as possible.

To evaluate the performance of the hybrid-objective QSD scheme, we conduct simulations with $\ell = 1$ for two test cases, each involving a different set of predefined states with equal prior probabilities. The first case uses the predefined states given in \eqref{eq:2-qubits states}, while the second case involves the two nonorthogonal single-qubit states:
\begin{equation}\label{eq:1-qubit states}
\ket{0} \quad \text{and} \quad \frac{1}{\sqrt{2}}\left(\ket{0}+\ket{1}\right).
\end{equation}
The corresponding results are shown in \Cref{fig:entangled_psucc,fig:nonortho_psucc}.
As before, the $L_2$ distance is employed to quantify the closeness between the noisy and noiseless probability distributions.

The performance remains robust and consistent over a broad range of relatively low noise levels, where both $P_{\mathrm{succ}}$ and the $L_2$ distance smoothly approach their respective MED values (horizontal dashed lines at the top) as $w$ decreases, and those of the noiseless UQSD (horizontal dotted lines at the bottom) as $w$ increases. As the noise level increases beyond a certain threshold, $P_{\mathrm{succ}}$ begins to degrade gradually, while the $L_2$ distance likewise departs from its low-noise behavior.\footnote{As the noise level $\lambda$ increases, all input states are progressively driven toward the maximally mixed state, i.e.,
$\mathcal{E}_\lambda(\rho_i)=\rho_i' \xrightarrow{\lambda \to 1} I/2^n$.
In this limit, the noisy states $\{\rho_i'\}$ become nearly identical and therefore essentially indistinguishable. Consequently, the MED (i.e., $w \to 0$ limit) POVM $\{\Pi_i\}$ reduces to the strategy of always guessing the state with the highest prior probability, and the corresponding success probability $P_{\mathrm{succ}}$ approaches $\max_i \{p_i\}$.
In contrast, the UQSD (i.e., $w \to \infty$ limit) POVM $\{\Pi_i\}$ reduces to the strategy of always reporting the inconclusive outcome, and $P_{\mathrm{succ}}$ thus approaches zero.}
Overall, these findings demonstrate that the hybrid-objective QSD scheme offers a reliable and robust continuous interpolation between MED and UQSD under moderate noise conditions.

\begin{figure}[b]
\centering
\includegraphics[width=0.48\columnwidth]{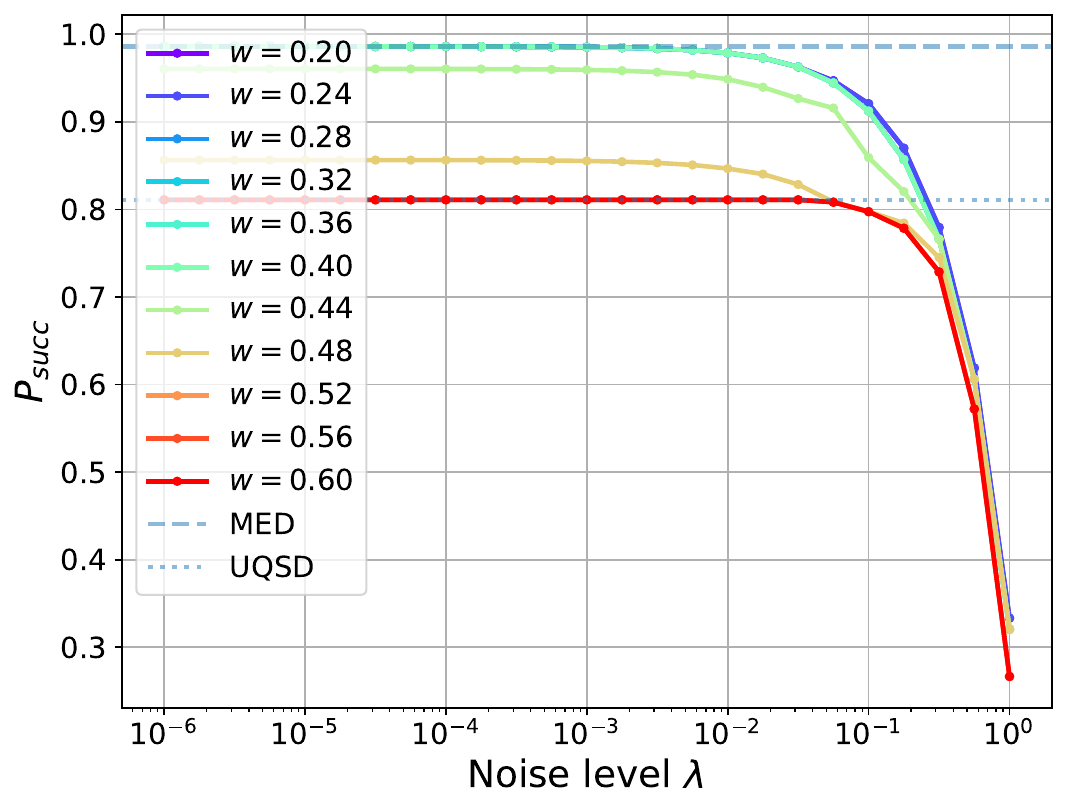}
\includegraphics[width=0.48\columnwidth]{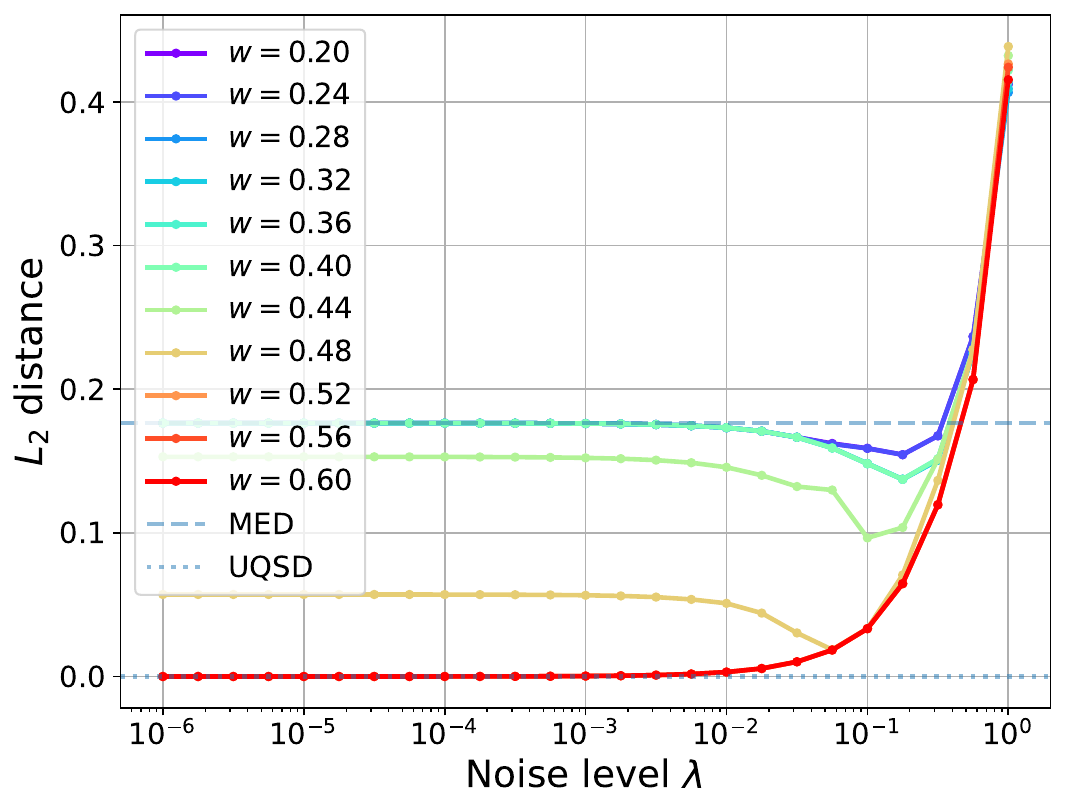}   
\caption{Hybrid-objective QSD with $\ell=1$ and the predefined states in \Cref{eq:2-qubits states}: Success probability $\psucc$ (left) and $L_2$ distance (right) vs.\ noise level $\lambda$.}
\label{fig:entangled_psucc}
\end{figure}
\begin{figure}[t]
\centering
\includegraphics[width=0.48\columnwidth]{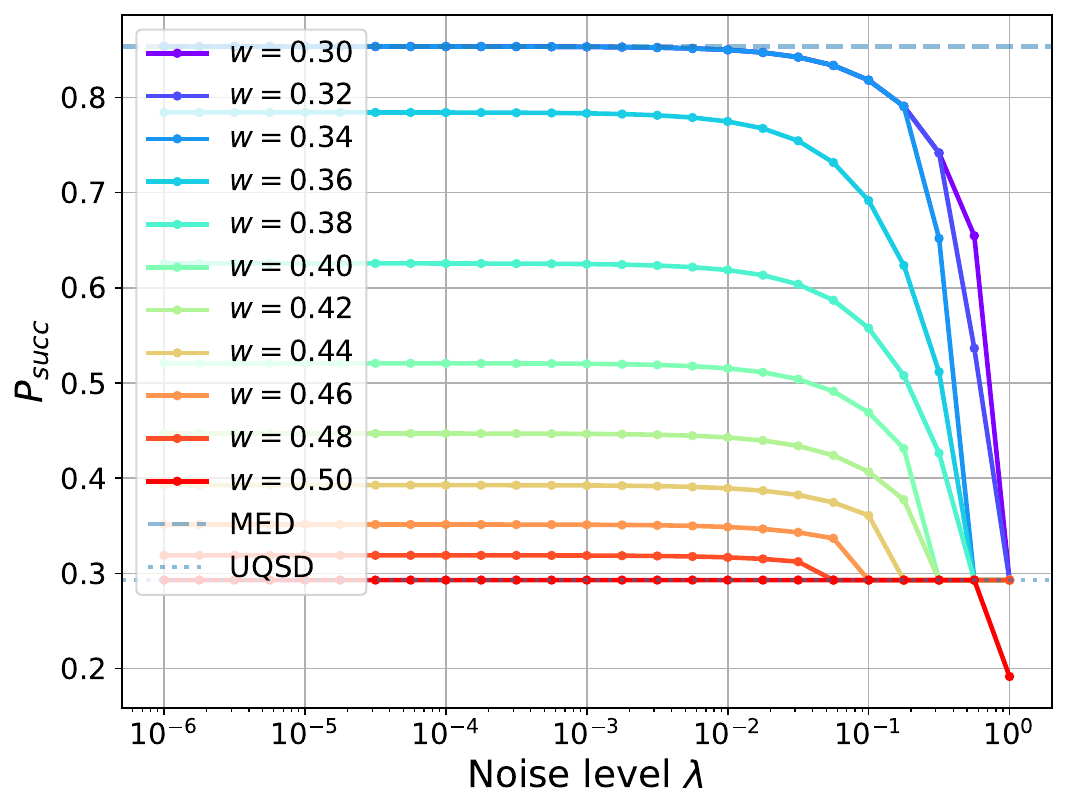}
\includegraphics[width=0.48\columnwidth]{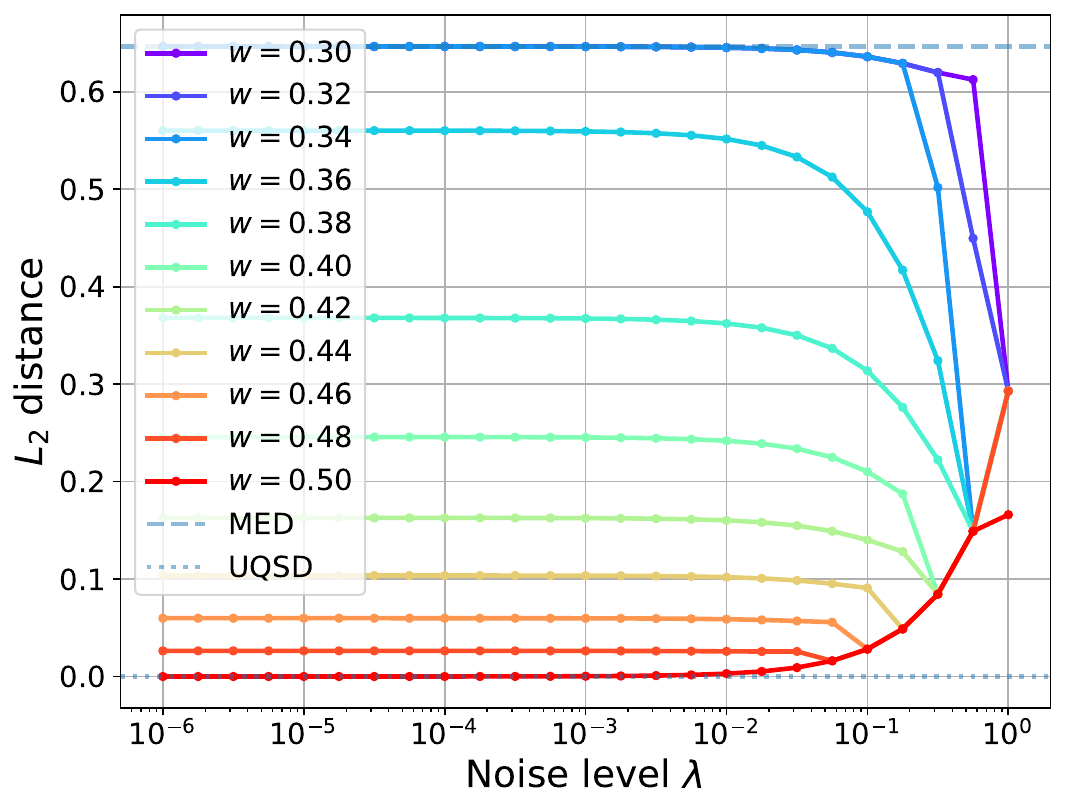}
\caption{Hybrid-objective QSD with $\ell=1$ and the predefined states in \Cref{eq:1-qubit states}: Success probability $\psucc$ (left) and $L_2$ distance (right) vs.\ noise level $\lambda$.}
\label{fig:nonortho_psucc}
\end{figure}


\section{Quantum Circuit Synthesis}
\label{sec:qc_qsd}

Given a POVM solution, we now turn to the problem of implementing a quantum circuit that realizes the corresponding measurement.
The synthesis of such a circuit involves two main steps:
(i) determining an \emph{isometry}---a norm-preserving linear map---that maps the POVM to a projection-valued measure (PVM) under specified constraints, and
(ii) constructing this isometry as an efficient quantum circuit that meets the desired resource and structural requirements.

Tailored specifically for QSD, our toolkit identifies the isometry that minimizes the number of required ancilla qubits using a novel approach, while the actual circuit-level realization of the isometry is handled by existing circuit-synthesis and circuit-optimization packages.
In the following, we first briefly review conventional circuit synthesis for POVMs and then introduce our new synthesis method designed for QSD, which leverages a modified version of Naimark's dilation theorem, as presented in \Cref{thm:min dilation}.

Let $\{\Pi_1, \Pi_2, \dots, \Pi_k\}$ be a POVM on a Hilbert space $\mathcal{H}$, and let \linebreak $\{P_1, P_2, \dots, P_k\}$ be a rank-1 PVM on an extended Hilbert space $\mathcal{H}' = \mathcal{H} \otimes \mathcal{H}_A$, where $\mathcal{H}_A$ is an ancillary system (dilation space). Naimark's dilation theorem guarantees the existence of an isometry
\begin{equation}\label{eq:def isometry}
 V: \mathcal{H}\rightarrow \mathcal{H}',   
\end{equation} 
satisfying
\begin{equation}\label{eq:VHV P}
V \Pi_i V^\dagger = P_i.   
\end{equation}
By choosing the rank-1 PVM in the form $P_i = I \otimes (\ket{i}_A \bra{i}_A)$, the corresponding isometry can then be explicitly constructed as
\begin{equation}\label{eq:dilation general}
V = \sum_{i=1}^k \sqrt{\Pi_i} \otimes \ket{i}_A,
\end{equation}
where $\sqrt{\,\cdot\,}$ denotes the matrix square root. Once $V$ is known, a quantum circuit can be synthesized using standard isometry circuit synthesis methods, such as the efficient one proposed in \cite{RN12}.
While this approach ensures that all POVM outcomes are encoded in the ancillary system and that the unmeasured subsystem yields the corresponding post-measurement states, it requires an ancillary space of at least dimension $k$, leading to $\operatorname{dim}\mathcal{H}' \geq k\operatorname{dim}\mathcal{H}$. This is generally wasteful for QSD, where the post-measurement states are largely irrelevant.


Instead of directly applying Naimark's dilation theorem to construct the isometry, we introduce the following modified version, whose proof is provided in \Cref{sec:proof-3}:
\begin{theorem}\label{thm:min dilation}
Let $\mathbf{\Pi} = \{\Pi_1, \Pi_2, \dots, \Pi_k\}$ be a POVM on a Hilbert space $\mathcal{H}_A$, and let $\mathbf{P} = \{P_1, P_2, \dots, P_\ell\}$ be a rank-1 PVM on another Hilbert space $\mathcal{H}_{A'}$.  
If the number of elements in $\mathbf{P}$ satisfies $\ell = \sum_{i=1}^k r_i$, where $r_i = \mathrm{rank}(\Pi_i)$, then there exists an isometry
\begin{equation}
V : \mathcal{H}_A \longrightarrow \mathcal{H}_{A'}
\end{equation}
such that
\begin{equation}\label{eq:tm2}
V \Pi_i V^\dagger = \sum_{j=1}^{r_i} P_{i,j}, \quad \text{for all } i = 1, \dots, k,
\end{equation}
where the rank-1 projectors in $\mathbf{P}$ are reindexed as
\begin{equation}
\{P_1, P_2, \dots, P_\ell\} = \{P_{1,1}, \dots, P_{1,r_1},\, P_{2,1}, \dots, P_{2,r_2},\, \dots,\, P_{k,1}, \dots, P_{k,r_k}\}.
\end{equation}
\end{theorem}

This construction allows us to map a POVM to a PVM via the isometry
\begin{equation}\label{eq:V new method in main text}
V = \sum_{i=1}^k \sum_{j=1}^{r_i} \ket{g_{i,j}}_{A'} |\tilde{f}_{i,j}\rangle_A,
\end{equation}
as defined in \Cref{eq:V new method}.
This requires only $\operatorname{dim}\mathcal{H}' \ge \ell$, which is in general substantially smaller than the more wasteful requirement $\operatorname{dim}\mathcal{H}' \ge k\operatorname{dim}\mathcal{H}$ in the previous approach based on \Cref{eq:dilation general}, i.e., typically, $\ell \ll  k\operatorname{dim}\mathcal{H}$ as $r_i \ll \operatorname{dim}\mathcal{H}$. Moreover, if the states $\ket{g_{i,j}}$ associated with $P_{i,j} = \ket{g_{i,j}}\bra{g_{i,j}}$ are chosen to be readily measurable---for instance, $\ket{g_{i,j}} = \ket{\sum_{k=1}^{i-1} r_k + j}$ as computational-basis states---then the POVM measurement can be effectively implemented as a PVM measurement.
Unlike the previous approach, however, each $\Pi_i$ now might correspond to multiple projectors $P_{i,j}$ for $j = 1, \dots, r_i$, and the post-measurement state is no longer available.

In the context of quantum circuits, by considering $\mathcal{H}_A \subset \mathcal{H}_{A'}$, the POVM can thus be systematically realized through the following procedure:
(i) embed any to-be-measured state $\rho \in \mathcal{D}(\mathcal{H}_A)$ into $\mathcal{D}(\mathcal{H}_{A'})$ by appending ancilla qubits;
(ii) apply the quantum gate that implements the isometry $V$;
(iii) perform a computational-basis measurement; and
(iv) whenever the measurement outcome associated with $P_{i,j}$ is obtained, record it as a $\Pi_i$ event.

We provide a toolkit that automatically identifies the isometry that requires the fewest ancilla qubits and generates the corresponding quantum circuit. In addition, the toolkit supports an approximation mode that achieves significantly higher circuit efficiency at the cost of a small, controllable precision loss.
Once a precision threshold\footnote{The approximation is intended to remove numerically insignificant components of the POVM operators that typically arise from numerical noise in the SDP solver. It is generally difficult to distinguish numerical artifacts from physically meaningful small components in an adaptive manner. For this reason, we expose the truncation threshold as a tunable parameter for the user.} $\delta$ is specified, the terms $|\tilde{f}_{i,j}\rangle$ in \Cref{eq:ldl} with $\sigma_{i,j} < \delta$ are discarded, yielding a truncated set of $|\tilde{f}_{i,j}\rangle$.
Removing these components in \Cref{eq:V new method in main text} renders $V$ no longer a strict isometry, but standard isometry-synthesis methods can still be applied to obtain an approximate realization based on the truncated set.
Although this approximation slightly reduces the overall gate fidelity, it offers a substantial advantage: the number of required ancilla qubits is further reduced, and both the gate count and circuit depth are significantly improved.

Once the isometry $V$ is determined, the toolkit synthesizes the corresponding quantum circuit at the abstract hardware level using the isometry-synthesis method~\cite{RN12} implemented in the \texttt{qclib} package~\cite{qclib}.
When circuit approximation is desired, the toolkit further applies the approximate quantum compiling technique introduced in~\cite{dc_gm}, which reduces circuit depth while maintaining high fidelity.
Finally, if a specific hardware topology is provided through Qiskit's backend class, the toolkit leverages Qiskit's \texttt{transpile} function to perform circuit resynthesis, generating a hardware-efficient implementation.
The overall synthesis workflow is illustrated in \Cref{fig:qc_workflow}.

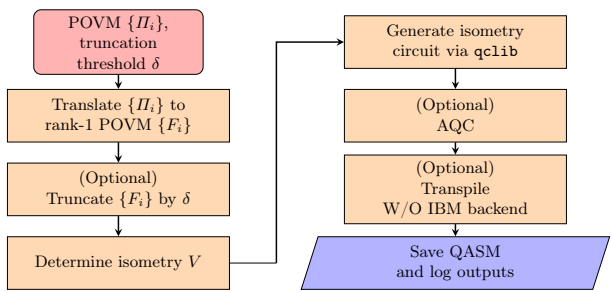
\begin{figure}[htbp]
\centering
\scalebox{0.7}{
        \usetikzlibrary{shapes.geometric, arrows.meta, positioning}

\tikzstyle{startstop} = [rectangle, rounded corners,
minimum width=3cm,
minimum height=1cm,
text centered,
text width=3cm,
draw=black,
fill=red!30]

\tikzstyle{io} = [trapezium,
trapezium stretches=true, 
trapezium left angle=70,
trapezium right angle=110,
minimum width=3cm,
minimum height=1cm,
align=center,
text width=5cm,
text centered,
draw=black, fill=blue!30]

\tikzstyle{process} = [rectangle,
minimum width=3cm,
minimum height=1cm,
text centered,
text width=4cm,
draw=black,
fill=orange!30]

\tikzstyle{decision} = [diamond,
minimum width=3cm,
minimum height=1cm,
text width=4cm,
aspect=2,
inner sep=-1ex,
text centered,
draw=black,
fill=green!30, 
aspect=2]
\tikzstyle{arrow} = [thick,->,>=stealth]

\begin{tikzpicture}[node distance= 0.4 cm and 0.2 cm]


    \node (start) [startstop] {POVM $\{\Pi_i\}$, \\truncation threshold $\delta$};
    \node(convert) [process, below of= start, yshift=-1cm, align=center] {Translate $\{\Pi_i\}$ to rank-1 POVM $\{F_i\}$};
    \node(truncate) [process, below of= convert, yshift=-1cm,align=center] {(Optional) \\Truncate $\{F_i\}$ by $\delta$};

    
    
    \node (iso) [process, below of= truncate, yshift=-1cm,align=center] {Determine isometry $V$};
    \coordinate (pt) at ($(iso) + (3, 0)$);
    \node (synth) [process, right of= start, xshift=6cm] {Generate isometry circuit via \texttt{qclib}};
    \node (approx) [process, below of= synth, yshift=-1cm] {(Optional)\\AQC};
    \node (transpile) [process, below of= approx, yshift=-1cm] {(Optional)\\Transpile\\ W/O IBM backend};
    \node (save) [io, below of= transpile, yshift=-1cm] {\shortstack{Save QASM \\ and log outputs}};

    \draw [arrow] (start) -- (convert);
    \draw [arrow] (convert) -- (truncate);
    \draw [arrow] (truncate) -- (iso);
    \draw (iso) -- (pt);
    \draw [arrow] (pt) |- (synth);
    \draw [arrow] (synth) -- (approx);
    \draw [arrow] (approx) -- (transpile);
    \draw [arrow] (transpile) -- (save);
\end{tikzpicture}
    }
\caption{Quantum circuit synthesis workflow for QSD.}
\label{fig:qc_workflow}
\end{figure}

We demonstrate a simple example of circuit synthesis using our toolkit, with the performance results summarized in \Cref{tab:med_circuit_size}. The predefined quantum states are specified in \Cref{eq:2-qubits states}, the prior probabilities are chosen uniformly as $p_1 = p_2 = p_3 = 1/3$, and the noise is turned off ($\lambda = 0$). The isometry corresponding to the optimal MED POVM $\{\Pi_i\}$ is first determined and subsequently approximated as $\{\tilde{\Pi}_i\}$ by truncating numerically insignificant components of the POVM operators according to a precision threshold $\delta = 10^{-4}$, using the optimization procedures provided by the toolkit. The resulting quantum circuit is then generated using the cosine--sine decomposition implemented in~\cite{qclib} without applying any additional circuit-level optimizations. 

As shown in the table, the total POVM rank is significantly reduced by the approximation (i.e., $\sum_i \mathrm{rank}(\tilde{\Pi}_i) < \sum_i \mathrm{rank}(\Pi_i)$). Correspondingly, the circuit complexity decreases substantially, while the QSD performance remains virtually unchanged: the differences in the outcome probabilities $p(\Pi_i|\rho_i)$ are negligible. This highlights a key advantage of the approximation scheme, which yields a considerable reduction in circuit complexity with almost no compromise in discrimination performance.

\begin{table}[ht]
    \centering
    \caption{Comparison of circuit complexities, total POVM ranks, and QSD outcome probabilities for the optimal MED, with and without approximation-based optimization (precision threshold $\delta = 10^{-4}$).}
    \label{tab:med_circuit_size}
    \renewcommand{\arraystretch}{1.2}
    \begin{tabular}{l | c@{\hskip 1.5mm}cc | c@{\hskip 1.9mm}cc}
        \hline
        Circuit specs & \multicolumn{3}{c|}{\textbf{Without Approximation}} & \multicolumn{3}{c}{\textbf{With Approximation}} \\
        \hline
        Ancilla qubits  & \multicolumn{3}{c|}{2}   & \multicolumn{3}{c}{1}  \\
        Single-qubit gates & \multicolumn{3}{c|}{150}  & \multicolumn{3}{c}{34} \\
        Two-qubit gates & \multicolumn{3}{c|}{68}  & \multicolumn{3}{c}{15} \\
        Circuit depth   & \multicolumn{3}{c|}{127} & \multicolumn{3}{c}{25} \\
        \hline\hline
        POVM total rank & \multicolumn{3}{c|}{12}  & \multicolumn{3}{c}{6} \\
        \hline\hline
         & \multicolumn{3}{c|}{Outcome probabilities}  & \multicolumn{3}{c}{Difference in outcome probabilities}\\
        Given state $\rho_i$ &
         $p(\Pi_{1}|\rho_i)$ & $p(\Pi_{2}|\rho_i)$ & $p(\Pi_{3}|\rho_i)$ &
         $p(\tilde{\Pi}_{1}|\rho_i)$ & $p(\tilde{\Pi}_{2}|\rho_i)$ & $p(\tilde{\Pi}_{3}|\rho_i)$ \\
         & & & & $-p(\Pi_{1}|\rho_i)$ & $-p(\Pi_{2}|\rho_i)$ & $-p(\Pi_{3}|\rho_i)$ \\
        \hline
        $\ket{\psi_1}$ in \Cref{eq:2-qubits states} & 0.99547 & 1.64E-3 &  2.89E-3 & $-$1.51E-12 & $-$3.61E-09 & $-$3.78E-09\\
        $\ket{\psi_2}$ in \Cref{eq:2-qubits states} & 1.66E-3 & 0.98188 & 1.65E-2 & $-$3.15E-09 &  $-$1.36E-11 & $-$1.86E-09\\
        $\ket{\psi_3}$ in \Cref{eq:2-qubits states} & 2.93E-3 & 1.65E-2 & 0.98059 & $-$3.03E-09 & $-$1.53E-09 & $-$2.08E-11\\
        \hline

    \end{tabular}
\end{table}

\color{black}
\section{Toolkit and Demonstration}
We present a lightweight Python toolkit that automatically generates quantum circuits for QSD problems according to predefined strategies. The toolkit integrates the SDP solver \emph{CVXPY}, an isometry matrix generator, and the isometry quantum circuit synthesis toolkit \emph{Qclib} within the \emph{Qiskit} framework. Its workflow comprises three stages: (i) construction of a QSD problem instance, specifying the quantum states with their prior probabilities and the noise model; (ii) synthesis of POVM operators according to predefined QSD strategies, which solves the POVM operators using \emph{CVXPY} and optionally provides approximate results; and (iii) quantum-circuit synthesis, which outputs the resulting \texttt{QuantumCircuit} object in \emph{Qiskit} through \emph{Qclib}.

In the following, we demonstrate a typical usage of our toolkit. The benchmarks for the computational efficiency of the functions in the toolkit are presented in Appendix~\ref{sef:performance}.


\subsection{Problem instance construction}
This step defines the basic configuration of the QSD problem using the \linebreak\texttt{ProblemSpec} class.
An example usage is shown below:
\begin{lstlisting}[language=Python]
from flow.solve_mix import *

# initialze the QSD problem instance
qsd_problem = ProblemSpec(num_qubits, num_states)

# set the prior probabilities by a numpy array
qsd_problem.prior_prob = np.ones(num_states) * (1 / num_states)

# specify the quantum states 
qsd_problem.set_states(state_type, states)

\end{lstlisting}

As shown in the example, the class initialization is intentionally designed to allow users to specify the number of quantum states to be discriminated and the number of qubits for each state before defining the actual quantum states.
When users need to replace the current set of quantum states, the \texttt{set\_states} function can be used, eliminating the need to reinitialize a new \texttt{ProblemSpec} object.
Additional information, such as prior probabilities, can also be specified later upon needed.

\subsection{POVM synthesis with an optimization strategy}
This step selects an optimization strategy and uses it to synthesize the POVM for the previously defined problem.
An example usage is shown below (continuing from the previous code block).
\begin{lstlisting}[language=Python]
from flow.build_circuits import *

crossqsd_prob = apply_crossQSD(
    qsd_problem,
    cvxpy_settings={"solver": cp.SCS, "verbose": False},
    alpha=[0.02] * qsd_problem.num_states,
    beta=[0.02] * qsd_problem.num_states
)

povm_vectors = get_povm_vectors(crossqsd_prob)
povm_ckt = POVMCircuit(povm_vectors=np.stack(povm_vectors))
povm_ckt.fix()
povm_ckt.build_circuit()
\end{lstlisting}

We provide a variety of commonly used optimization strategies in the \linebreak\texttt{flow.solve\_mix} package, which are summarized in \Cref{tab:solve_mix}.
The optimization process is implemented using CVXPY to obtain the optimal POVM corresponding to the selected strategy.
When using an optimization function, users must specify the target \texttt{ProblemSpec} instance, the configuration of the CVXPY solver, and any additional constraints required by the chosen optimization strategy, as described in \Cref{sec:error_uqsd} and listed in \Cref{tab:solve_mix}.

    \begin{table}[!h]
        \caption{Python function interfaces for the QSD formulations in the \texttt{flow.solve\_mix} package}
        \label{tab:solve_mix}
        \renewcommand{\arraystretch}{1.2}
        \setlength{\tabcolsep}{3pt}
        \centering
        \begin{tabular}{c|c|c|c}
            \hline
            function name & state type & QSD & extra arguments \\
            \hline
            \texttt{apply\_Eldar} & state vector & Optimal UQSD & \texttt{beta} \\
            \texttt{med\_problem} & density matrix & MED & (None) \\
            \texttt{med\_plus\_problem} & density matrix & \;\;MED$^{+}$ & (None) \\
            \texttt{apply\_frio} & density matrix & FRIO & \texttt{p\_inc\_lb} \\
            \texttt{apply\_crossQSD} & density matrix & CrossQSD & \texttt{alpha}, \texttt{beta}\\
            \texttt{min\_l1\_problem} & density matrix & FitQSD-minL1 & \texttt{ideal\_distrib} \\
            \texttt{min\_ss\_problem} & density matrix & FitQSD-minSS & \texttt{ideal\_distrib} \\
            \texttt{meco\_problem} & density matrix & FitQSD-MECO & \texttt{ideal\_distrib} \\
            \texttt{hybrid\_obj\_problem} & density matrix & Hybrid-objective & \texttt{ideal\_distrib, param\_a} \\
            \hline
        \end{tabular}
    \end{table}

After executing a selected optimizer (e.g., the \texttt{apply\_crossQSD} function in the example), a \texttt{POVMCircuit} instance can be created from the POVM vectors produced by the optimizer.
To enable the use of fewer qubits in circuit construction, users can explicitly specify the \texttt{num\_qubits} attribute.
Finally, the \texttt{fix()} function completes the POVM by adding residual components above a small singular-value threshold.
    
\subsection{Quantum circuit synthesis and resynthesis}
If the user requires a gate-level implementation, the compilation core can be invoked to generate the corresponding quantum circuit. 
An example usage is shown below (continuing from the previous code block).
\begin{lstlisting}[language=Python]
from flow.resynth import resynth_aqc
qc = povm_ckt.build_circuit("ccd")
qc = resynth_aqc(qc)
\end{lstlisting}
The \texttt{build\_circuit} function supports three synthesis schemes, \texttt{ccd}, \texttt{csd}, and \texttt{knill}, and internally employs the \texttt{qclib} package to generate the corresponding quantum circuits.
Note that when using the \texttt{ccd} scheme, the circuit is first extended to a full unitary before decomposition,spending additional computing time. 

Once the quantum circuit is generated, the estimated number of CNOT gates is displayed to the console for immediate review. 
The synthesized circuit is returned as a Qiskit \texttt{QuantumCircuit} instance for further use, and the provided \texttt{resynth\_aqc()} function enables approximate circuit resynthesis via Qiskit’s \texttt{transpile()} interface.


\section{Case Study: UQSD for Coherent States in Noisy Quantum Circuit Simulators}\label{sec:case study}
In this section, we demonstrate the use of our toolkit on quantum computers and simulate the QSD experiment using Qiskit's noisy quantum circuit simulator. For a fair performance assessment, we adopt the UQSD scheme with optimal POVM operators as the discrimination strategy. 

In the experiments, we consider three coherent states defined in \Cref{sec:crossqsd}, with parameters $\alpha = 1, e^{i \frac{2\pi}{3}}, e^{i \frac{4\pi}{3}}$. Since coherent states are intrinsically infinite-dimensional, they must be truncated to a finite Hilbert space for quantum circuit implementation. In the following processes, those coherent states are truncated into subspaces ranging from two to six qubits. The corresponding phase-space representations of the truncated states are shown in Figure~\ref{fig:coh_comparison}, which illustrates how the overlap among the three states diminishes as the truncation dimension increases.
\begin{figure}[htbp]
    \centering
    
    \subfloat[\#qubits = 3\label{fig:coh_q3_n3}]{%
    \includegraphics[width=.24\textwidth]{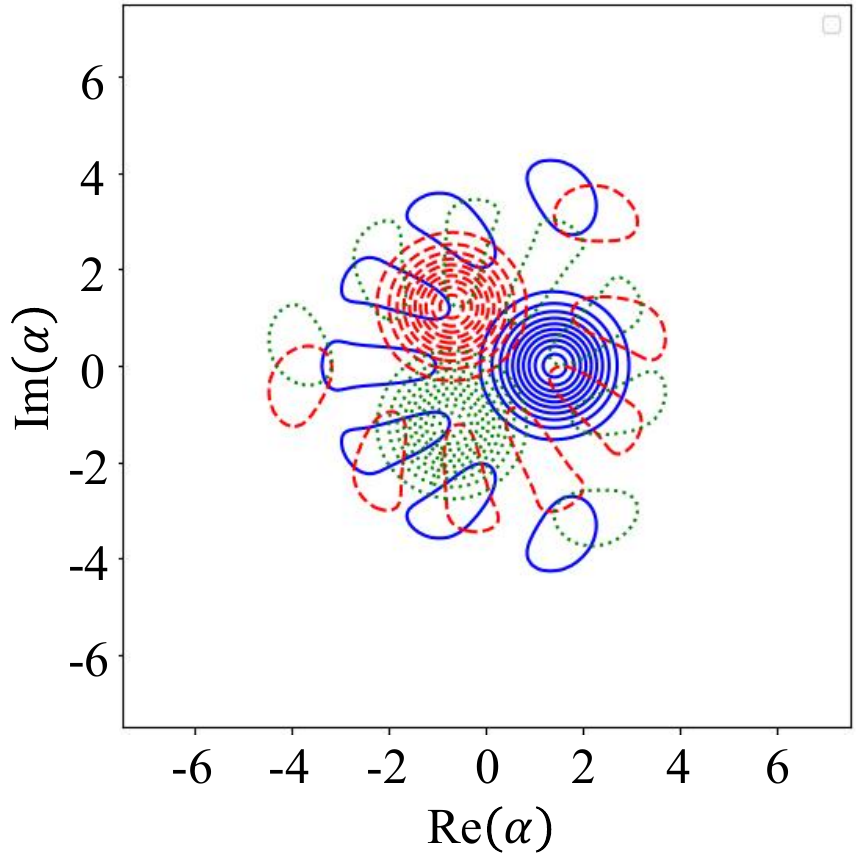}
    }\hfill
    \subfloat[\#qubits = 4\label{fig:coh_q4_n3}]{%
    \includegraphics[width=.24\textwidth]{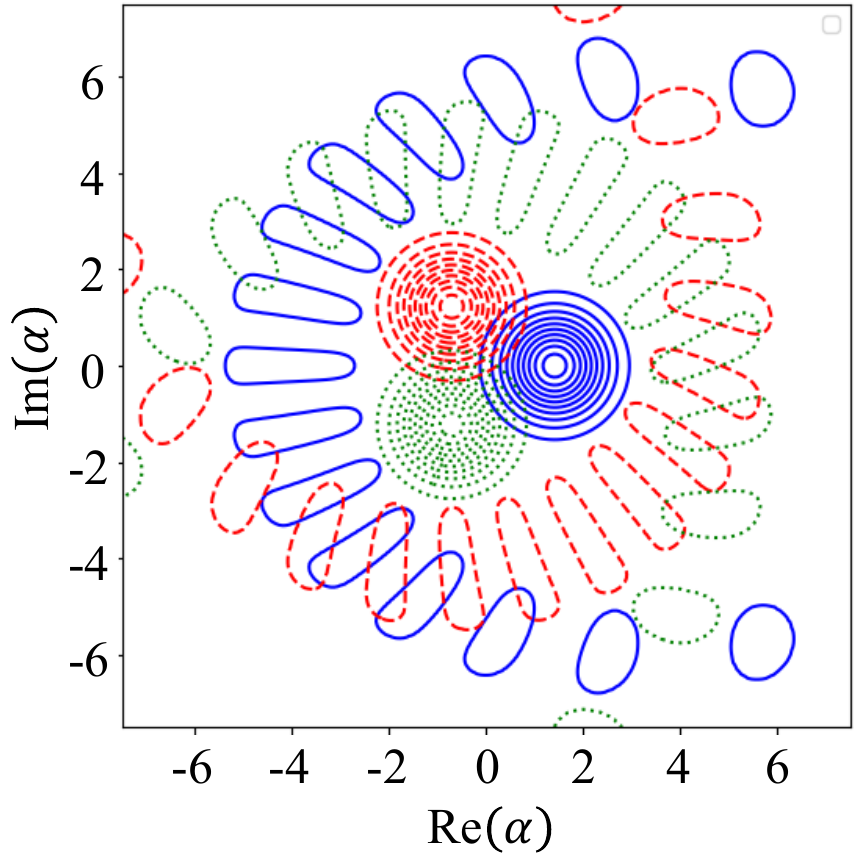}%
    }\hfill
    \subfloat[\#qubits = 5\label{fig:coh_q5_n3}]{%
    \includegraphics[width=.24\textwidth]{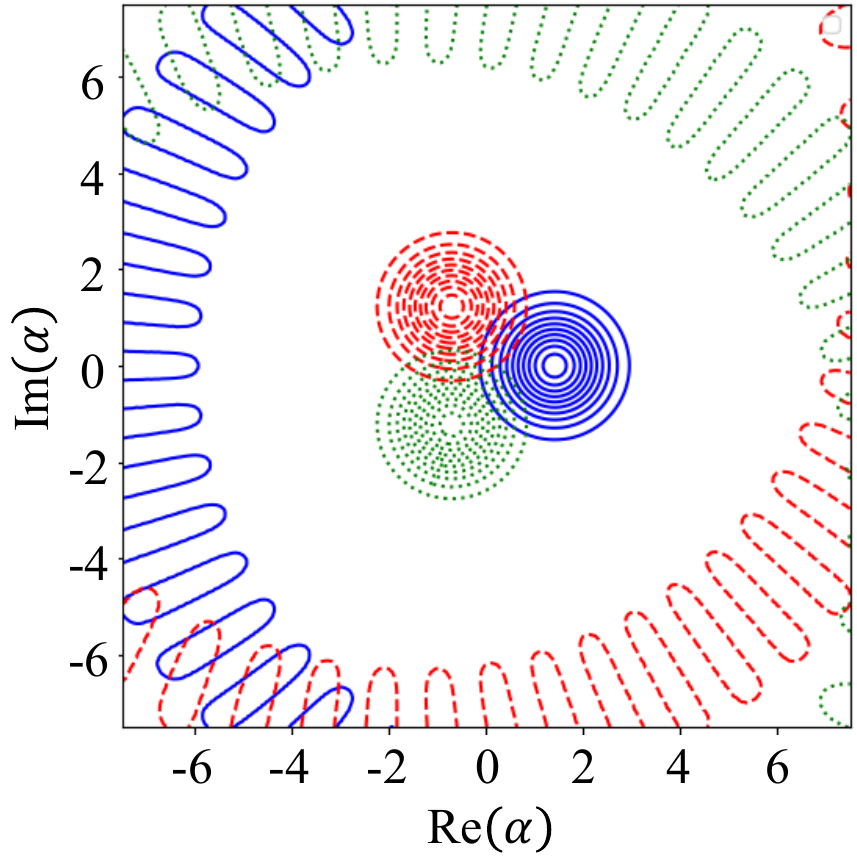}%
    }\hfill
    \subfloat[\#qubits = 6\label{fig:coh_q6_n3}]{%
    \includegraphics[width=.24\textwidth]{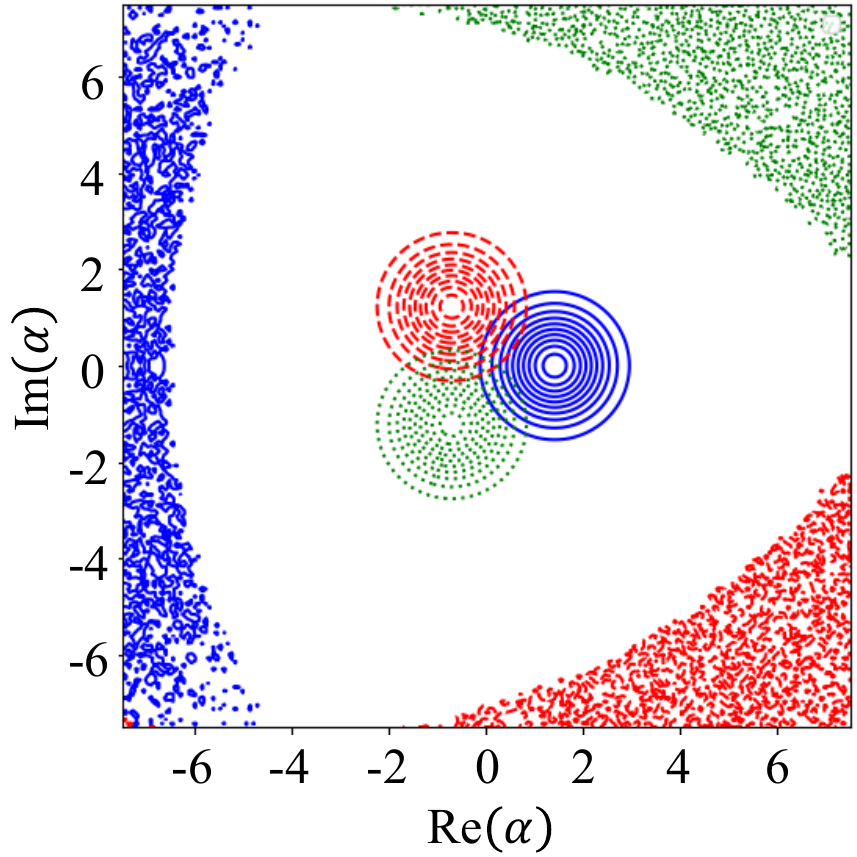}%
    }

    \caption{Wigner functions of truncated coherent states represented with 3 to 6 qubits.
    The blue, green, and red dots correspond to the states $\ket{\alpha}$ with
    $\alpha = 1, e^{i \frac{2\pi}{3}},$ and $e^{i \frac{4\pi}{3}}$, respectively.
    All figures are generated using QuTiP \cite{lambert2024qutip5quantumtoolbox}.}
  
    \label{fig:coh_comparison}
\end{figure}


For the QSD implementation, the optimal POVM operators are first obtained using \texttt{apply\_Eldar}, after which the corresponding quantum circuits are synthesized with \texttt{build\_circuit("ccd")}. The resulting circuits are then resynthesized into approximate forms using \texttt{approx\_circuit}, which adopts the AQC optimization settings from Qiskit’s \texttt{transpile()} function. The average circuit depth, two-qubit gate count, and state fidelity before and after approximation are summarized in \Cref{tab:synth_comparison} for the three representative states. Notably, no ancillary qubits are required in the gate-level synthesis, as the total ranks of the POVM operators in this example are equal to three.

    \begin{table}[htbp]
    \centering
    \caption{Comparison of quantum circuit synthesis methods for
        discriminating three coherent states.
        Here \#2Q means the number of two-qubit gates, which is
        also the number of CNOT gates in this case.
        The numbers for the state fidelity and process fidelity are rounded down to the fifth digit.}
    \label{tab:synth_comparison}
    \scalebox{0.63}{
        \renewcommand{\arraystretch}{1.2}
        \setlength{\tabcolsep}{3pt}
        \begin{tabular}{|c|c||rr||rrr|rrr|rrr|rrr|}
            \hline
            \multirow{3}{*}{\textbf{\#Qubit}}
              & \multicolumn{1}{c||}{\multirow{3}{*}{\tabincell{c}{\textbf{State}\\\textbf{Fidelity}}}}
              & \multicolumn{2}{c||}{\textbf{Original}}
              & \multicolumn{3}{c|}{\textbf{Resynthesis}}
              & \multicolumn{3}{c|}{\textbf{Approx\_deg = 1.00}}
              & \multicolumn{3}{c|}{\textbf{Approx\_deg = 0.97}}
              & \multicolumn{3}{c|}{\shortstack{\textbf{AQC (\#2Q = 10)}}}                                                                                                                                                       \\\cline{3-16}
              &                                                          & Depth & \#2Q & Depth & \#2Q & \tabincell{c}{Process\\Fidelity} & Depth & \#2Q & \tabincell{c}{Process\\Fidelity} & Depth & \#2Q & \tabincell{c}{Process\\Fidelity} & Depth & \#2Q & \tabincell{c}{Process\\Fidelity} \\
            \hline
            2 & 0.97803                                                  & 84    & 41    & 45    & 20   & 1.00000          & 29    & 14   & 0.99999          & 29    & 14   & 0.59409          & 21    & 10 & 0.35564          \\
            3 & 0.99998                                                  & 430   & 218   & 225   & 100  & 0.99999          & 83    & 61   & 0.99992          & 83    & 61   & 0.88925          & 15    & 10 & 0.96252          \\
            4 & 0.99999                                                  & 2034  & 1025  & 1009  & 444  & 0.99999          & 253   & 252  & 0.99989          & 253   & 252  & 0.13754          & 11    & 10 & 0.32505          \\
            5 & 0.99999                                                  & 8916  & 4474  & 4273  & 1868 & 0.99999          & 817   & 1020 & 0.99990          & 817   & 1020 & 0.00032          & 9     & 10 & 0.25978          \\
            6 & 0.99999                                                  & 37220 & 18634 & 17585 & 7660 & 0.99999          & 2729  & 4091 & 0.99984          & -     & -    & -                & -     & -  & -                \\
            \hline
        \end{tabular}
    }
\end{table}

Finally, we perform the QSD protocol for the six-qubit truncated states using Qiskit’s noisy quantum circuit simulator, \texttt{Qiskit\_Aer}. The simulator is configured to emulate depolarization noise on every two-qubit gate, defined as 
\begin{align}
\mathcal{E}(\rho^{(2)}) = (1 - \lambda) \rho^{(2)} + \lambda \frac{I^{(2)}}{2^n},
\end{align}
where $\lambda$ denotes the level of depolarizing noise, $n$ is the total number of qubits in the circuit, and $\rho^{(2)}$ is the two-qubit reduced density matrix that after gate application. Simulations are performed for various $\lambda$ values ranging from $10^{-6}$ to $10^{-1}$, with each configuration executed using $1024$ shots. The measurement outcomes are subsequently post-processed to estimate the empirical success probabilities.

The experimental results are presented in Figure~\ref{fig:noisy_simulation}. The success probabilities for discriminating each truncated six-qubit coherent state remain close to the noiseless values for $\lambda < 10^{-3}$, but drop sharply when $\lambda > 10^{-3}$. These findings confirm our earlier observation that UQSD fails immediately once the noise exceeds a critical threshold.

\begin{figure}[h]
\centering
\includegraphics[width=0.85\textwidth]{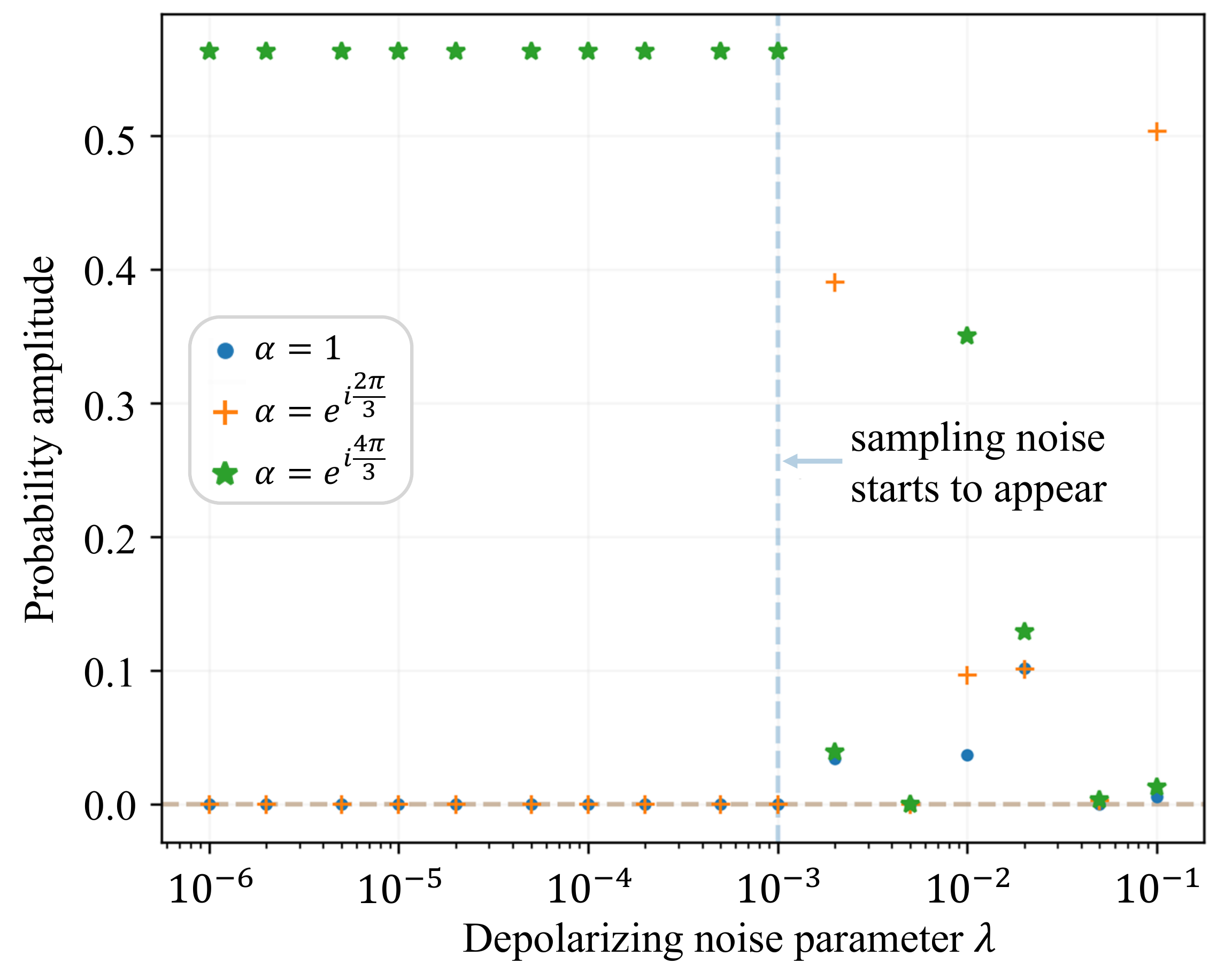}
\caption{Experimental success probabilities for discriminating each six-qubit-truncated coherent states vs. noise level $\lambda$. 
}
\label{fig:noisy_simulation}
\end{figure}

\section*{Acknowledgments}

This work was supported in part by the National Science and Technology Council of Taiwan under grants 114-2119-M-002-020, the NTU Center of Data Intelligence: Technologies, Applications, and Systems under grant NTU-114L900903, and the NTU Core Consortium Project NTU-CC114L895002.

\section*{Data Availability Statement}

The source code and benchmark data used in this study are publicly available at
\url{https://zenodo.org/records/17493285}.

\newpage
\bibliographystyle{splncs04}
\bibliography{reference.bbl}

\appendix
\crefalias{subsection}{appendix}
\section{Proofs}

\subsection{Proof of \Cref{thm:med_medplus}}
\label{sec:proof-1}
\begin{proof}
For the given $\{\rho_{1}, \rho_{2}, \dots, \rho_{k}\}$ and $\{p_{1}, p_{2}, \dots, p_{k}\}$, let $\mathcal{K}$ and $\mathcal{K}^+$ denote the collections of all feasible POVMs that satisfy the constraints of MED and MED${}^+$, respectively, regardless of optimality.
For any $K \equiv \{\Pi_1, \Pi_2, \dots, \Pi_k\} \in \mathcal{K}$, the corresponding success probability is
\begin{equation}
(\psucc)^K = \sum_{i=1}^k p_i \Tr(\rho_i \Pi_i).
\end{equation}
Similarly, for any $K^+ \equiv \{\Pi_1, \Pi_2, \dots, \Pi_k, \Pi_{k+1} \equiv \Pi_?\} \in \mathcal{K}^+$, the corresponding success probability is
\begin{equation}
(\psucc)^{K^+} = \sum_{i=1}^k p_i \Tr(\rho_i \Pi_i),
\end{equation}
where the inconclusive operator $\Pi_?$ does not contribute to $\psucc$ by definition.  
The optimal success probabilities are then given by
\begin{equation}
(\psucc)\supmed = \max_{K \in \mathcal{K}} (\psucc)^K, \qquad
(\psucc)\supmedplus = \max_{K^+ \in \mathcal{K}^+} (\psucc)^{K^+}.
\end{equation}

Let $\bar{K} \equiv \{\bar{\Pi}_1, \bar{\Pi}_2, \dots, \bar{\Pi}_k\} \in \mathcal{K}$ denote the optimal MED POVM.
Construct the corresponding MED${}^+$ POVM
\begin{equation}
\bar{K}^+ \equiv \{\bar{\Pi}_1, \bar{\Pi}_2, \dots, \bar{\Pi}_k, \bar{\Pi}_? = 0\} \in \mathcal{K}^+.
\end{equation}
Since including a null operator does not affect the success probability, we have
\begin{equation}
(\psucc)^{\bar{K}^+} = (\psucc)^{\bar{K}} = (\psucc)\supmed.
\end{equation}
As $\bar{K}^+$ is a feasible (though not necessarily optimal) MED${}^+$ POVM, it follows that
\begin{equation}\label{ineq:a}
(\psucc)\supmedplus \geq (\psucc)^{\bar{K}^+} = (\psucc)\supmed.
\end{equation}

Conversely, let $\tilde{K}^+ \equiv \{\tilde{\Pi}_1, \tilde{\Pi}_2, \dots, \tilde{\Pi}_k, \tilde{\Pi}_?\} \in \mathcal{K}^+$ denote the optimal MED${}^+$ POVM.  
Construct a corresponding MED POVM as
\begin{equation}
\tilde{K} \equiv \{\tilde{\Pi}_1, \tilde{\Pi}_2, \dots, \tilde{\Pi}_{k-1},\tilde{\Pi}_k + \tilde{\Pi}_?\} \in \mathcal{K}.
\end{equation}
The associated success probability is then
\begin{equation}
(\psucc)^{\tilde{K}}=  \sum_{i=1}^{k} p_i \operatorname{Tr}\left(\rho_i \tilde{\Pi}_i \right) + p_k \operatorname{Tr}\left(\rho_k \tilde{\Pi}_? \right)
\geq (\psucc)\supmedplus,
\end{equation}
where equality holds if and only if $\tilde{\Pi}_? = 0$.  
Since $\tilde{K}$ is a feasible MED POVM (though not necessarily optimal), we have
\begin{equation}\label{ineq:b}
(\psucc)\supmed \geq (\psucc)^{\tilde{K}} \geq (\psucc)\supmedplus.
\end{equation}

Combining \Cref{ineq:a} and \Cref{ineq:b} yields
\begin{equation}
(\psucc)\supmed = (\psucc)\supmedplus,
\end{equation}
which implies that the optimal MED${}^+$ solution must have $\Pi_? = 0$, completing the proof.
\end{proof}

\subsection{Proof of \Cref{corollary}}
\label{sec:proof-2}

\begin{proof}
Both MED${}^+$ and UQSD can be formulated as convex optimization problems sharing the same objective function. The only distinction lies in their constraints: UQSD additionally imposes the no-misidentification condition: $\Tr(\rho_i \Pi_j) = 0$ for all $i \neq j$ with $i,j = 1,\dots,k$. Consequently, the collection of POVMs satisfying the UQSD constraints, regardless of optimality, forms a subset of that satisfying the MED${}^+$ constraints, i.e.,
\begin{equation}\label{inclusion rel}
\mathcal{K}_\mathrm{UQSD} \subseteq \mathcal{K}_{\mathrm{MED}^+}.
\end{equation}
Since both optimization problems maximize the same linear functional but over different convex feasible spaces, $\mathcal{K}_\mathrm{UQSD}$ and $\mathcal{K}_{\mathrm{MED}^+}$, respectively, the optimal value of MED${}^+$ cannot be smaller than that of UQSD because of the inclusion relation \Cref{inclusion rel}.  
Finally, invoking \Cref{thm:med_medplus}, we obtain
\begin{equation}
(\psucc)\supmed = (\psucc)\supmedplus \geq (\psucc)\supuqsd,
\end{equation}
which completes the proof.
\end{proof}

\subsection{Proof of \Cref{thm:min dilation}}

\label{sec:proof-3}
\begin{proof}
By the singular value decomposition (or equivalently, the LDL decomposition), each $\Pi_i$ can be expressed as a sum of rank-1 positive-definite operators:
\begin{equation}\label{eq:ldl}
\Pi_i
= L_i \Sigma_i L_i^\dagger
= \sum_{j=1}^{r_i} \sigma_{i,j} \ket{l_{i,j}}\bra{l_{i,j}}
\equiv \sum_{j=1}^{r_i} |\tilde{f}_{i,j}\rangle\langle\tilde{f}_{i,j}|
= \sum_{j=1}^{r_i} F_{i,j},
\end{equation}
where $\Sigma_i = \sum_{j=1}^{r_i} \sigma_{i,j} \ket{j}\bra{j}$ (that is, $\sigma_{i,j}$ is the $j$th singular value), $\ket{l_{i,j}} = L_i \ket{j}$, $|\tilde{f}_{i,j}\rangle = \sqrt{\sigma_{i,j}} \ket{l_{i,j}}$, and $F_{i,j} = |\tilde{f}_{i,j}\rangle\langle\tilde{f}_{i,j}|$.\footnote{The tilde in $|\tilde{\psi}\rangle$ indicates that the state may not be normalized.}

Since each $P_{i,j}$ is rank-1, it can be written as $P_{i,j} = \ket{g_{i,j}}\bra{g_{i,j}}$.  
We then define the isometry
\begin{equation}\label{eq:V new method}
V = \sum_{i=1}^k \sum_{j=1}^{r_i} \ket{g_{i,j}}_{A'} |\tilde{f}_{i,j}\rangle_A,
\end{equation}
which immediately yields
\begin{equation}\label{eq:VFV}
V F_{i,j} V^\dagger = P_{i,j}.
\end{equation}
Combining \Cref{eq:ldl} and \Cref{eq:VFV}, we obtain
\begin{equation}
V \Pi_i V^\dagger = \sum_{j=1}^{r_i} P_{i,j}.
\end{equation}

Finally, since $\{\ket{g_{i,j}}\}$ form an orthonormal basis as $\mathbf{P}$ is a PVM, and
\begin{equation}
\sum_{i=1}^k \Pi_i
= \sum_{i=1}^k \sum_{j=1}^{r_i} |\tilde{f}_{i,j}\rangle\langle\tilde{f}_{i,j}|
= I_A,
\end{equation}
it follows that
\begin{equation}
V^\dagger V = I_A,
\end{equation}
confirming that $V$ is norm-preserving and therefore an isometry.
\end{proof}

\section{Performance benchmarks}\label{sef:performance}
To assess the computational efficiency of the functions provided by our toolkit, we benchmark their runtime in relation to the number of qubits of the predefined states.

For each QSD scheme, we time the following three tasks:
\begin{enumerate}
    \item \emph{Task I}: Given a QSD instance, solve the corresponding optimization problem to obtain the POVM $\{\Pi_i\}$.
    \item \emph{Task II}: Convert the POVM $\{\Pi_i\}$ into an equivalent rank-1 POVM $\{F_i\}$.
    \item \emph{Task III}: Determine the isometry $V$ from $\{F_i\}$ and synthesize the corresponding isometry circuit using \texttt{qclib}.
\end{enumerate}

All benchmarks are performed on a machine equipped with two Intel Xeon Silver 4314 CPUs (2.40~GHz, max 64 threads) and 128~GB of RAM. We test systems ranging from 2 to 6 qubits; the cases of more than 6 qubits are not included because the available memory on our test machine is insufficient to complete the computation.

The exact settings used for these benchmarks are listed in \Cref{tab:benchmark parameters}, and the benchmark results are presented in \Cref{fig:Task I}--\Cref{fig:Task III}.

    \begin{table}[!h]
        \caption{The predefined states, prior probabilities, and specific user-defined parameters used for the benchmarks in different QSD schemes as presented \Cref{fig:Task I}--\Cref{fig:Task III}.}
        \label{tab:benchmark parameters}
        \renewcommand{\arraystretch}{1.2}
        \setlength{\tabcolsep}{3pt}
        \centering
        \begin{tabular}{c|c|c}
            \hline
            $\rho_i$ and $p_i$ & QSD scheme & User-defined parameters \\
            \hline
            \multirow{9}{4.7cm}{Predefined states:\\ \hspace{0.3cm} The three truncated coherent states corresponding to $\alpha = 1, e^{2\pi i/3}$, and $e^{4\pi i/3}$ as described in \Cref{sec:case study}.\\ \vspace{0.3cm} Prior probabilities:\\ \hspace{0.5cm} $p_1=p_2=p_3=1/3$.} & Optimal UQSD & (None) \\
             & MED & (None) \\
             & MED$^{+}$ & (None) \\
             & FRIO & Fixed Rate of IO: $0.1$ \\
             & CrossQSD & $\alpha_{1,2,3}=0.1$, $\beta_{1,2,3}=0.1$\\
             & FitQSD-minL1 & $\ell=1$ \\
             & FitQSD-minSS & $\ell=2$ \\
             & FitQSD-MECO & (None) \\
             & Hybrid-objective & $\ell=1$, $w=0.3$ \\
            \hline
        \end{tabular}
    \end{table}

    \begin{figure}
        \centering
        \includegraphics[width=0.75\linewidth]{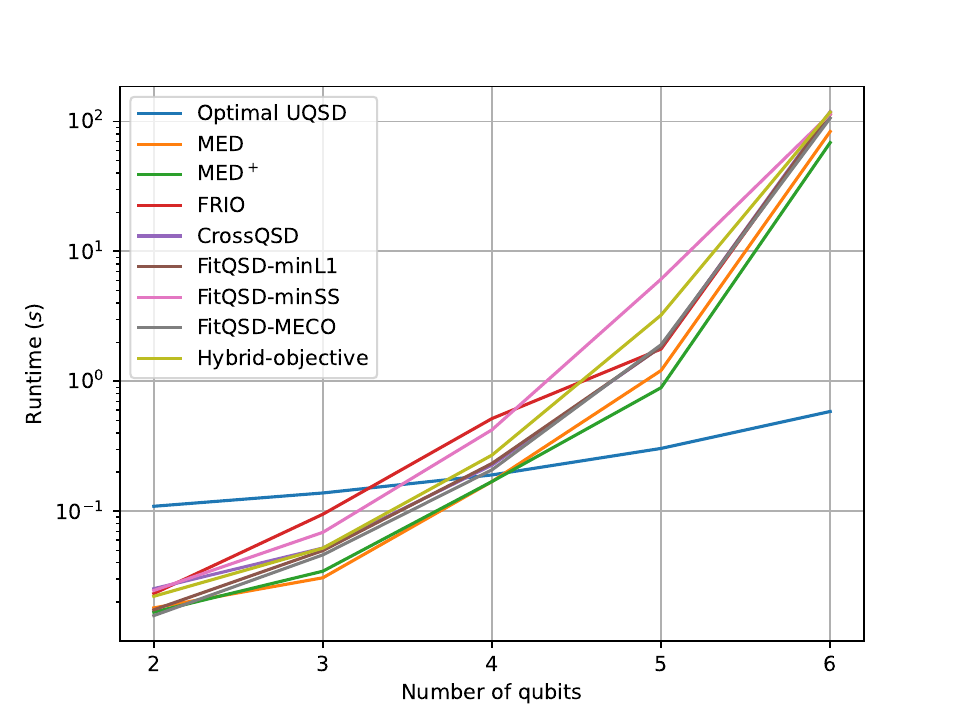}
        \caption{Runtime benchmarks for Task~I.}
        \label{fig:Task I}
    \end{figure}
    \begin{figure}
        \centering
        \includegraphics[width=0.75\linewidth]{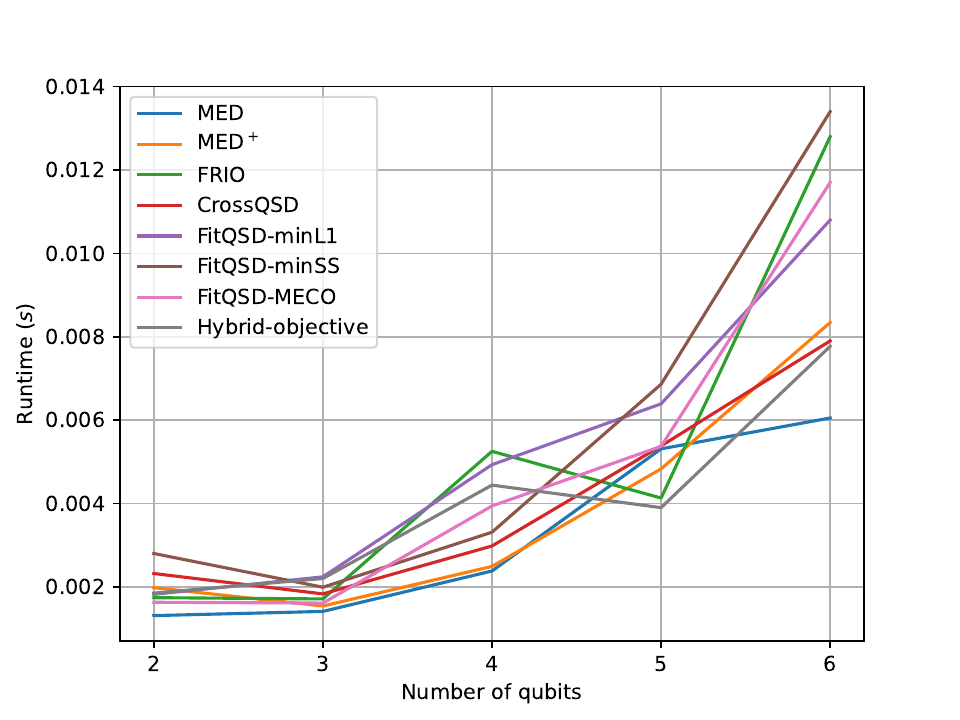}
        \caption{Runtime benchmarks for Task~II.}
        \label{fig:Task II}
    \end{figure}
    \begin{figure}
        \centering
        \includegraphics[width=0.75\linewidth]{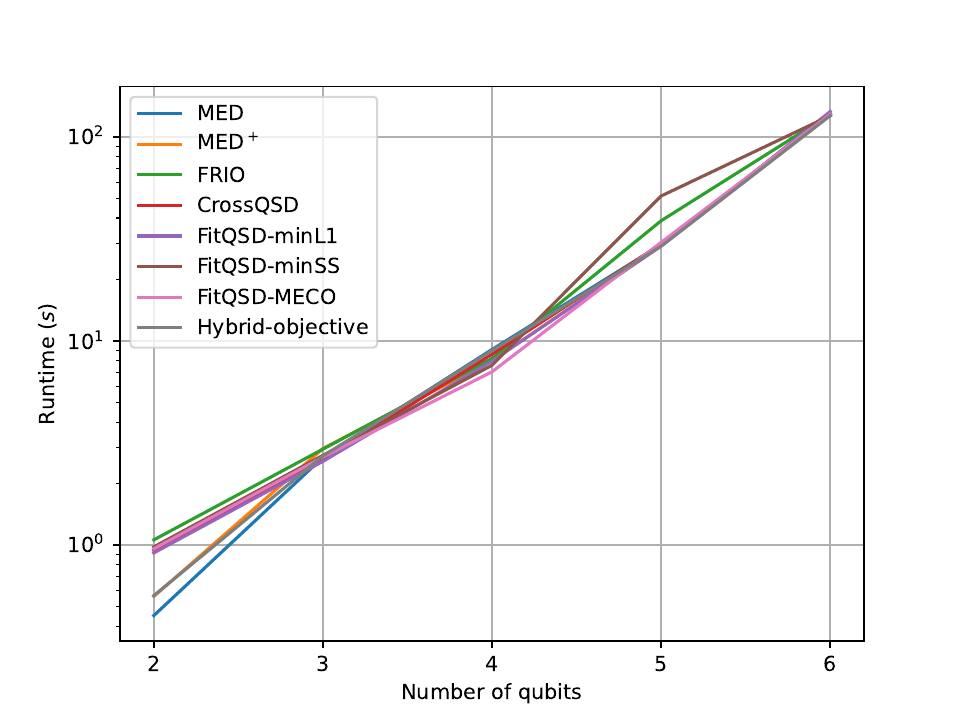}
        \caption{Runtime benchmarks for Task~III.}
        \label{fig:Task III}
    \end{figure}

\end{document}